\documentclass{article}
\usepackage[utf8]{inputenc}
\usepackage{dsfont}
\usepackage{amsthm}
\usepackage{amsmath}
\usepackage{caption}
\usepackage{subcaption}
\usepackage{graphicx}
\usepackage[ruled,vlined]{algorithm2e}
\usepackage[longnamesfirst]{natbib}
\usepackage{pgfplots}
\usepackage[title]{appendix}
\pgfplotsset{width=6.5cm,compat=1.5}
\DeclareMathOperator{\LN}{LN}
\DeclareMathOperator{\BS}{BS}
\DeclareMathOperator{\Put}{P}

\usepackage{nicefrac}

\usepackage{xcolor}

\newtheorem{theorem}{Theorem}
\newtheorem{remark}{Remark}
\newtheorem{proposition}{Proposition}
\newtheorem{corollary}{Corollary}
\newtheorem{lemma}{Lemma}

\title{Higher order approximation of option prices in Barndorff-Nielsen and Shephard models}
\author{{\'A}lvaro Guinea Juli{\'a} and Alet Roux}

\begin{document}

\maketitle

\begin{abstract}
   We present an approximation method based on the mixing formula \citep{hull1987pricing,romano1997contingent} for pricing European options in Barndorff-Nielsen and Shephard models. This approximation is based on a Taylor expansion of the option price. It is implemented using a recursive algorithm that allows us to obtain closed form approximations of the option price of any order (subject to technical conditions on the background driving L{\'e}vy process). This method can be used for any type of Barndorff-Nielsen and Shephard stochastic volatility model. Explicit results are presented in the case where the stationary distribution of the background driving L{\'e}vy process is inverse Gaussian or gamma. In both of these cases, the approximation compares favorably to option prices produced by the characteristic function. In particular, we also perform an error analysis of the approximation, which is partially based on the results of \citet{das2022closed}. We obtain asymptotic results for the error of the $N^{\text{th}}$ order approximation and error bounds when the variance process satisfies an inverse Gaussian Ornstein–Uhlenbeck process or a gamma Ornstein–Uhlenbeck process. 
\end{abstract}

\section{Introduction}

In this paper, we present a new approximation method for pricing European options in Barndorff-Nielsen and Shephard stochastic volatility models\break \citep{Barndorff-Nielsen_Shephard2001a,Barndorff-Nielsen_Shephard2001b}. \citet{nicolato2003option} show that it was possible to construct a closed formula for pricing European options in Barndorff-Nielsen and Shephard models, using the characteristic function. For this type of model, estimation results from historical data showed that the parameter representing the mean reversion rate could have large values \citep{hubalek2011joint,GuineaJulia2022}. When the mean reversion rate  is high, the characteristic function becomes numerically unstable (see \citet[Section~4.5]{GuineaJulia2022}). Because of that, it is convenient to have an approximation formula for pricing European options.

A first approximation result for option prices in Barndorff-Nielsen and Shephard models was given by \citet{schroder2014contingent}, who used Laguerre series to construct an approximation for option prices when there is no correlation parameter included in the model. \citet{arai2022approximate} proposed another approximation method, which showed that for these models it is possible to construct an Al\`os type formula \citep{alos2012decomposition}, which is a generalisation of the Hull-White formula \citep{hull1987pricing} that allows for correlation parameter. The approximation method presented in this paper follows a different approach. We construct an approximation formula using the Romano and Touzi formula \citep{romano1997contingent} and then we use the $N^{\text{th}}$ order Taylor expansion to approximate option prices. We construct a bound for the error using the results of \citet{das2022closed}, who built an error bound for a second order approximation of the Heston model \citep{heston1993closed} based on the Romano and Touzi formula. The bound given by \citet{das2022closed} for the high-order derivatives that appear on Taylor's remainder is stochastic; see \citet{Das2024Erratum}. We prove that these high-order derivatives can be bounded by a deterministic function when the volatility is bounded from below, as in the Barndorff-Nielsen and Shephard model. Furthermore, we generalize these results to $N^{\text{th}}$ order approximations and adapt them to the Barndorff-Nielsen and Shephard model. This allows us to get asymptotic results for the error approximation and error bounds when the variance process follows an inverse Gaussian Ornstein–Uhlenbeck process or a gamma Ornstein–Uhlenbeck process.

This paper is organised as follows. Firstly, we introduce the Barndorff-Nielsen and Shephard model in Section \ref{PreSect}. The approximation formula is derived in Section \ref{SectApprox}. This approximation formula depends on several moments of the integrated variance process. We build a recursive procedure to compute these moments in Section \ref{SectMoments}. The approximation error is discussed in Section \ref{ErrorSect}. Numerical results are shown in Section \ref{NumExamplesSect}, where we compare the numerical values given by our approximation method with the prices given by the characteristic function. In Appendix \ref{DerBoundsect}, we construct a deterministic bound for the high-order derivatives of the error term. Analytical formulas for the characteristic functions are given in Appendix \ref{CFsect}.

\section{Preliminaries}\label{PreSect}

Consider the Barndorff-Nielsen and Shephard model with a fixed time horizon $T>0$ \cite[cf.][]{Barndorff-Nielsen_Shephard2001b,Barndorff-Nielsen_Shephard2001a}. Take as given an equivalent martingale measure $\mathds{Q}$ in the class of structure-preserving equivalent martingales introduced by \citet[Theorem 3.2]{nicolato2003option}. The price process $S=(S_t)$ of the stock is defined as
\[
 S_t = e^{X_t} \text{ for all }t\in[0,T],
\]
where the dynamics of the log price process $X=(X_t)$ and the stochastic volatility process $\sigma=(\sigma_t)$ under $\mathds{Q}$ are given by
\begin{align}
 dX_t &= \left(r-\lambda \kappa(\rho) - \tfrac{1}{2}\sigma_t^2\right)dt + \sigma_tdW_t + \rho dZ_{\lambda t}, & X_0 &= x_0, \label{eq:X:SDE} \\
 d\sigma^2_t &= -\lambda \sigma^2_t dt + dZ_{\lambda t}, & \sigma^2_0 &>0.\label{eq:Sigma:SDE}
\end{align}
Here $r\in\mathds{R}$ is the interest rate, $\lambda>0$ is the mean reversion rate, $W = (W_t)$ is a Brownian motion, $Z= (Z_t)$ is a (L\'evy) subordinator process independent of $W$ and $\kappa$ is the cumulant generating function of $Z_1$, in other words, 
\[\kappa(\theta)=\ln E_{\mathds{Q}}\left[e^{\theta Z_1}\right]\text{ for all }\theta\in\mathds{R}.\] 
We define
\[
\hat{\kappa}=\sup\{\theta\in\mathds{R}:\kappa(\theta)<\infty\},
\]
assume that $\hat{\kappa}>0$ and that the correlation parameter $\rho\in\mathds{R}$ satisfies
\begin{equation}\label{eq:rho_cond}
     \rho<\hat{\kappa}.
\end{equation} This is a relaxation of the usual condition $\rho\le0$.

Since $Z$ is a non-negative L\'evy process, it has finite total variation on a bounded interval  \citep[Proposition 3.10]{cont2004financial} and the cumulant generating function of $Z_1$ can be written as
\[\kappa(\theta) =  b_Z +  \int_0^\infty \left(e^{\theta x} -1\right) v(dx)  \text{ for } u \in \mathds{R},\]  where $b_Z\geq 0$ and $v$ is the L\'evy measure of $Z$. By the L\'evy-Khintchine representation \citep[p.~37]{sato1999levy},   the process $Z$ has L\'evy triplet $(\gamma,0, v)$ where
\[\gamma = b_Z + \int_{|x|\leq 1} x v(dx)\] \citep[Corollary 3.1]{cont2004financial}.

The process $\sigma^2$ is an Ornstein-Uhlenbeck process with stationary distribution $D$. These processes are usually called $D$-Ornstein-Uhlenbeck processes. This type of processes are built from self-decomposable distributions $D$\break \citep{valdivieso2009maximum}. Examples of self-decomposable distributions are the gamma distribution and the inverse Gaussian distribution \citep[Section 5.5]{schoutens2003levy}. It is possible to show the existence of an Ornstein-Uhlenbeck process with stationary distribution $D$ if the distribution $D$ is self-decomposable \citep{valdivieso2009maximum}. In this paper we use two types of $D$-Ornstein-Uhlenbeck processes as examples, namely the gamma  Ornstein-Uhlenbeck process and the inverse Gaussian Ornstein-Uhlenbeck process.

The solution of the stochastic differential equation \eqref{eq:Sigma:SDE} can be written as \[\sigma_t^2 = e^{-\lambda t} \left( \sigma_0^2 +\int_0^t e^{\lambda s} dZ_{\lambda s}\right) \text{ for all } t\geq 0.\] The distribution of the random variable $\sigma_t^2$ will depend on the self-decomposable distribution $D$. Define the integrated variance as
\[
 I_t = \int_0^t \sigma_s^2ds \text{ for all }t\ge0.
\]
We have\begin{equation} \label{eq:I-SDE}
    I_t = \sigma^2_0 \alpha_{0,t} +  \int_0^t \alpha_{s,t} dZ_{\lambda s}
\end{equation}  where \begin{equation}\label{eq:alpha_fun}
    \alpha_{s,t} =  \frac{1}{\lambda}(1-e^{-\lambda(t-s)}) \text{ for all } s,t\geq 0.
\end{equation} \citep[(2.5)]{nicolato2003option}.

Let $(\mathcal{F}^W_t)$ and $(\mathcal{F}^Z_t)$ be the filtrations generated by $W$ and $Z$, respectively, and define the filtration $(\mathcal{F}_t)$ as
\[
 \mathcal{F}_t = \mathcal{F}^W_t \vee \mathcal{F}^Z_{\lambda t} \text{ for all } t\geq 0.
\]

It is essential for the approximation below to observe that the stochastic differential equation \eqref{eq:X:SDE} can be separated into terms connected with $W$ and $Z$, respectively. Defining the auxiliary process $P=(P_t)$ as 
\[
 P_t = e^{\rho Z_{\lambda t} - \lambda t\kappa(\rho)} \text{ for all }t\ge0,
\]
we obtain
\begin{align*}
 X_t &= x_0 + rt - \tfrac{1}{2}I_t + \int_0^t\sigma_sdW_s + \ln P_t, \\
 S_t &= S_0P_te^{rt - \frac{1}{2}I_t + \int_0^t\sigma_sdW_s},
\end{align*}
for all $t\geq 0$, which means that the distribution of $S_t$ conditional on $\mathcal{F}^Z_{\lambda t}$ is lognormal, in other words,
\begin{equation} \label{eq:St-given-FZlt}
 S_t|\mathcal{F}^Z_{\lambda t}\sim \LN\left(\ln S_0P_t + rt - \tfrac{1}{2}I_t,I_t\right) \text{for all } t>0.
\end{equation}

\section{Approximation}\label{SectApprox}

We focus on the payoff of a European put with expiration date $T>0$ and strike $K\ge0$, the reason being that the payoff is bounded. The price at time $0$ of the put option under $\mathds{Q}$ is
\begin{align*}
\Pi_{\Put} (S_0,K,T)
&= E_{\mathds{Q}}\left[ e^{-rT}  (K-S_T)^+\right] \\
&= E_{\mathds{Q}}\left[ E_{\mathds{Q}}\left[e^{-rT} (K-S_T)^+\middle|  \mathcal{F}_{\lambda T}^Z  \right] \right],
\end{align*} where $\Pi_P(S_0,K,T)$ represents the put price in the Barndorff-Nielsen and Shephard model with initial value of the stock price $S_0$, strike price $K$ and expiration date $T$.

The distribution of $S_T$ conditional on $\mathcal{F}_{\lambda T}^T$ is lognormal by \eqref{eq:St-given-FZlt}, and therefore standard arguments lead to a representation in terms of the familiar Black-Scholes-Merton formula, namely
\[
    E_{\mathds{Q}}\left[e^{-rT} (K-S_T)^+\middle|  \mathcal{F}_{\lambda T}^Z  \right] = \BS_{\Put}\left(S_0P_T, I_T \right),
\]
where
\begin{align}
  \BS_{\Put}(x,y) &= K e^{-rT} \Phi(-d_-(x,y)) - x\Phi(-d_+(x,y)), \label{eq:BSPut} \\
  d_{\pm}(x,y) &= \frac{1}{\sqrt{y}} \left(\ln\frac{x}{K} + rT \pm \tfrac{1}{2}y\right), \label{eq:d+-}
\end{align}
and $\Phi$ is the cumulative distribution function of the standard normal distribution. Hence the value at time $0$ of the put option is:
\begin{equation}\label{eq:Put-Price}
    \Pi_{\Put}(S_0,K,T) = E_{\mathds{Q}}\left[  \BS_{\Put}\left(S_0P_T, I_T \right) \right].
\end{equation}
This representation dates back to the work of \citet{romano1997contingent} (and \citet{hull1987pricing} in the case $\rho=0$).

The aim is to derive an approximation to \eqref{eq:Put-Price} using a Taylor series expansion of $\BS_{\Put}$ around the point 
$\left(S_0, E_{\mathds{Q}}[I_T] \right)$. To this end, we have the following result.

\begin{theorem}\label{ApproxThm}
    For any $N\in\mathds{N}$, if 
\begin{equation}\label{MomentEq}
     E_{\mathds{Q}}\left[\lvert P_T-1\rvert^{n-k} \left\lvert I_T- E_{\mathds{Q}}[I_T] \right\rvert^k \right] < \infty
\end{equation} for all $n=1,\ldots,N+1$ and $k=0,\ldots,n$, then $\Pi_{\Put}$ of \eqref{eq:Put-Price} can be approximated as
\begin{equation} \label{eq:Taylor-approx}
 \Pi_{\Put}(S_0,K,T) = \Pi_N + R_N  \approx \Pi_N,
\end{equation}
where the \emph{approximation} is
\begin{multline} \label{eq:Taylor-approx-N}
 \Pi_N = \BS_{\Put}(S_0,E_{\mathds{Q}}[I_T]) \\
 + \sum_{n=2}^N \frac{1}{n!} \sum_{k=0}^n \binom{n}{k} S_0^{n-k}E_{\mathds{Q}}\left[(P_T-1)^{n-k} \left(I_T - E_{\mathds{Q}}[I_T]\right)^k\right] \frac{\partial^n\BS_{\Put}}{\partial x^{n-k} \partial y^k} (S_0,E_{\mathds{Q}}[I_T])
\end{multline}and the \emph{remainder term} is
\begin{multline} \label{eq:Taylor-remainder-N}
 R_N
 = \frac{1}{N!}\sum_{n=0}^{N+1} \binom{N+1}{n}S_0^{N+1-n}E_{\mathds{Q}}\left[(P_T-1)^{N+1-n} \left(I_T - E_{\mathds{Q}}[I_T]\right)^n \phantom{\int_0^1}\right. \\
 \left.\times\int_0^1(1-u)^N \frac{\partial^{N+1}\BS_{\Put}}{\partial x^{N+1-n} \partial y^n} ((1-u)S_0+uP_TS_0,(1-u)E_{\mathds{Q}}[I_T]+uI_T)du\right].
\end{multline}
\end{theorem}

\begin{proof}
 Using the $N^\text{th}$ order Taylor series with remainder in integral form \citep[see, for example,][chapter 6]{Duistermaat2010}, we obtain
\begin{multline*} \BS_{\Put}\left(S_0 P_T, I_T \right) = \BS_{\Put}(S_0,E_{\mathds{Q}}[I_T]) \\
 \begin{aligned}
 &+ \sum_{n=1}^N \frac{1}{n!} \sum_{k=0}^n \binom{n}{k} (S_0P_T-S_0)^{n-k} \left(I_T - E_{\mathds{Q}}[I_T]\right)^k \frac{\partial^n\BS_{\Put}}{\partial x^{n-k} \partial y^k} (S_0,E_{\mathds{Q}}[I_T]) \\
 &+ \frac{1}{N!}\sum_{n=0}^{N+1} \binom{N+1}{n}(S_0P_T-S_0)^{N+1-n} \left(I_T - E_{\mathds{Q}}[I_T]\right)^n \\
 &\quad\times\int_0^1(1-u)^N \frac{\partial^{N+1}\BS_{\Put}}{\partial x^{N+1-n} \partial y^n} ((1-u)S_0+uS_0P_T,(1-u)E_{\mathds{Q}}[I_T]+uI_T)du.
 \end{aligned}
\end{multline*} It is well known (see, for example, Proposition \ref{PropDer} in Section \ref{ErrorSect}) that the derivatives of $\BS_{\Put}$ of all orders are well defined as long as they are evaluated at positive points. The derivatives that appear in the approximation are well defined since $S_0>0$ and $E_{\mathds{Q}}[I_T]>0$. The derivatives that appear in the error term, are well defined almost surely because
\begin{eqnarray}
    (1-u)S_0+uS_0P_T > 0, \\
    (1-u)E_{\mathds{Q}}[I_T]+uI_T>0,
\end{eqnarray}
 with probability $1$, for all $u \in(0,1)$.
The approximation \eqref{eq:Taylor-approx}--\eqref{eq:Taylor-remainder-N} is obtained after taking the expected value under $\mathds{Q}$.
\end{proof}

The integrability condition \eqref{MomentEq} as well as a method for computing the moments in \eqref{eq:Taylor-approx-N}, will be covered in Section \ref{SectMoments}. We end this section by explicitly stating the second order approximation.
\begin{corollary}
The second order approximation for the European put option is \begin{align}
\Pi_2 &=  \BS_{\Put}(S_0,E_{\mathds{Q}}[I_T])  \nonumber\\
& \qquad + \frac{1}{2} \frac{\partial^2 BS_{\Put} }{\partial x^2} \left(S_0,E[I_T]\right) S_0^2E_{\mathds{Q}}\left[(P_T -1)^2\right] \nonumber\\
& \qquad + \frac{1}{2} \frac{\partial^2 BS_{\Put} }{\partial y^2} \left(S_0,E[I_T]\right)  E_{\mathds{Q}}\left[\left( I_T - E_{\mathds{Q}}[I_T]\right)^2\right] \nonumber\\
& \qquad + \frac{\partial^2 BS_{\Put} }{\partial x \partial y} \left(S_0,E[I_T]\right) S_0 E_{\mathds{Q}}\left[ (P_T -1) (I_T - E[I_T])\right].\label{2ordereq}
\end{align} where \begin{align}
    E_{\mathds{Q}}\left[ I_T \right] & = \alpha_{0,T} \left(\sigma_0^2 - \kappa'(0)\right) + \kappa'(0) T, \label{eq:Exp_IT}\\
E_{\mathds{Q}}\left[ (I_T -E_{\mathds{Q}} [I_T])^2 \right] & = \frac{1}{\lambda^2} k''(0) \left( \lambda T - \frac{3}{2} + 2e^{-\lambda T}   - \frac{1}{2}e^{-2\lambda T}\right), \nonumber\\
E_{\mathds{Q}}\left[(P_T -1)^2\right] 
& = e^{\lambda T \left( - 2\kappa(\rho) + \kappa( 2\rho) \right) }  -1, \nonumber\\
E_{\mathds{Q}}\left[ (P_T -1) (I_T - E[I_T])\right] 
& = \left(\kappa'(\rho) -\kappa'(0)\right) \left( T  - \alpha_{0,T} \right). \nonumber
\end{align} where $\kappa$ is the cumulant generating function of $Z_1$.
\end{corollary}
\begin{proof}
   Equation \eqref{eq:Exp_IT} is given by \citet[p. 289]{barndorff2003integrated}, while the rest of the moments can be computed from Corollary \ref{cor:Moments} in Section \ref{SectMoments}.
\end{proof}

\section{Moments of price and integrated volatility process} \label{SectMoments}

The approximation \eqref{eq:Taylor-approx-N} contains mixed central moments of the form 
\begin{equation}\label{Moment1}
   E_{\mathds{Q}}\left[(P_T-1)^{n-k} \left(I_T- E_{\mathds{Q}}[I_T] \right)^k \right] 
\end{equation}
where $n\in\mathds{N}$ and $k=0,1,\ldots,n$. These moments can be calculated trivially in some cases, for example, they are equal to zero when $k=1$ or $n-k=1$. Another special case, namely moments that are not mixed, occurs when $k=0$ or $n=k$. In this section, we develop a method for computing these moments for all $k$ and $n$, whenever they are well defined. But before computing these moments, we will construct a recursive formula that we will use to calculate the moments.
\begin{lemma} \label{theo:MomentsR}
Let the process $F=(F_t)$ be defined as
\begin{equation}\label{eq:Rt}
    F_t =  f(t) +  \int_0^t (c \alpha_{s,t} + d) dZ_{\lambda s} \text{ for all }  t\geq 0, 
\end{equation} where $f: [0,\infty) \to \mathds{R}$ is a continuous deterministic function and $c,d\in\mathds{R}$. 

For any $c,d\geq 0$ with  $\max\{c,d\}>0$, any $\ell\in\mathds{N}_0$ such that $\ell\rho < \hat{\kappa}$ and any $k\in\mathds{N}_0$ such that $\kappa$ is $k$ times continuously differentiable in an open interval containing $0$, we have
\[
  E_{\mathds{Q}}\left[P_T^\ell F_T^k \right] = e^{\lambda T (\kappa(\ell \rho)-\ell \kappa(\rho))} H_{\ell,k},
 \]
 where $H_{\ell,k}$ satisfies the recursive relationship
  \begin{align}
  H_{\ell,0} &= 1, \nonumber\\
  H_{\ell,h} &= f(T) H_{\ell,h-1}  +  \lambda \sum_{i=1}^h \binom{h-1}{i-1} H_{\ell,h-i} 
\kappa^{(i)} (\ell \rho) \int_0^T  \left( c \alpha_{s,T}+d\right)^i ds \label{eq:LemmaH}
  \end{align}
 for all $h=1,\ldots,k$.
\end{lemma} \begin{proof}
 Define the process $\bar{P}=(\bar{P}_t)$ as
 \[
  \bar{P}_t = e^{\ell\rho Z_{\lambda t} - \lambda t \kappa(\ell \rho)} \text{ for all }t\ge0;
 \]
 then $\bar{P}$ is a $\mathds{Q}$-martingale with respect to $(\mathcal{F}_t)$ \cite[Theorem 13.50]{pascucci2011pde}. Define the probability measure $\bar{\mathds{Q}}$ on $\mathcal{F}_T$ by means of the Radon-Nikodym density
 \begin{equation} \label{eq:dQl/dQ}
  \frac{d\bar{\mathds{Q}}}{d\mathds{Q}} = \bar{P}_T;
 \end{equation}
 then
 \[
    E_{\mathds{Q}}\left[P_T^\ell F_T^k \right] = e^{\lambda T (\kappa(\ell \rho)-\ell \kappa(\rho))} E_{\mathds{Q}}\left[\bar{P}_T F_T^k \right] = e^{\lambda T (\kappa(\ell \rho)-\ell \kappa(\rho))} E_{\bar{\mathds{Q}}}\left[F_T^k \right].
\]
Define 
\[
 H_{\ell,h} = E_{\bar{\mathds{Q}}}\left[F_T^h \right] \text{ for }h=0,1,\ldots,k;
\]
then $H_{\ell,0}=1$. For $h>0$, we have
\begin{equation} \label{eq:H-ito-Mderiv}
 H_{\ell,h} = \bar{M}^{(h)}_{F_T}(0),
\end{equation}
provided that the moment generating function $\bar{M}_{F_T}$ of $F_T$ under $\bar{\mathds{Q}}$ is $h$ times continuously differentiable in an open interval containing $0$.

Define $\varepsilon=\tfrac{1}{2}\left(\frac{\hat{\kappa}}{\ell\rho}-1\right)$; then $\ell\rho < \frac{\hat{\kappa}}{1+\varepsilon}$ and moreover
\begin{equation}\label{MGFCond}
    \theta (c \alpha_{s,T} + d ) + \ell\rho \leq \frac{\hat{\kappa}}{1+\epsilon} \text{ for all }s\in[0,T]
\end{equation}
as long as $ \theta\le \hat{\theta}$, where\[
\hat{\theta}=\frac{\frac{\hat{\kappa}}{1+\varepsilon}-\ell\rho}{c \alpha_{0,T} + d} >0 
\] because $c,d\geq 0$ and  $\max\{c,d\}>0$. For any $\theta\le\hat{\theta}$, it follows from \eqref{eq:Rt} and \eqref{eq:dQl/dQ} that
\begin{align*}
    \bar{M}_{F_T}(\theta)
    &= E_{\bar{\mathds{Q}}} \left[ e^{\theta F_T} \right]\\
    & =  e^{\theta f(T)} E_{\bar{\mathds{Q}}}\left[e^{\theta \int_0^T  (c \alpha_{s,T} + d) dZ_{\lambda s}}\right] \\
    &= e^{\theta f(T) - \lambda T \kappa(\ell \rho)} E_{\mathds{Q}}\left[e^{ \ell\rho Z_{\lambda T} + \theta \int_0^T (c \alpha_{s,T} + d)  dZ_{\lambda s}}\right] \\
   &= e^{\theta f(T) - \lambda T \kappa(\ell \rho)} E_{\mathds{Q}}\left[e^{\int_0^T \left(\theta (c \alpha_{s,T} +d) +  \ell\rho\right) dZ_{\lambda s}}\right] \\
   &= e^{\theta f(T) - \lambda T \kappa(\ell \rho) + \lambda\int_0^T \kappa\left( \theta (c \alpha_{s,T} + d) + \ell\rho\right) ds}
   \end{align*}
 where the last equality is due to a result by \citet[Lemma~2.1]{nicolato2003option}.

We now show by induction that the $h^\text{th}$ derivative of $\bar{M}_{F_T}$ satisfies
\begin{multline} \label{eq:deriv-M-induction}
    \bar{M}_{F_T}^{(h)}(\theta) = f(T) \bar{M}_{F_T}^{(h-1)}(\theta)\\ 
    +\lambda\sum_{i=1}^h \binom{h-1}{i-1} \bar{M}_{F_T}^{(h-i)}(\theta) 
\int_0^T \kappa^{(i)} \left(  \theta(c \alpha_{s,T} + d) + \ell\rho \right) \left( c \alpha_{s,T} + d\right)^i ds
\end{multline}
for all $\theta<\hat{\theta}$. The first derivative of $\bar{M}_{F_T}$ can be computed directly using the chain rule and the Leibniz rule as
\begin{align*}
    \bar{M}'_{F_T}(\theta) 
 & =  f(T) \bar{M}_{F_T}(\theta) \\
&  \quad + \lambda \bar{M}_{F_T}(\theta) \int_0^T \kappa' \left(  \theta (c \alpha_{s,T} +d)  + \ell\rho \right) ( c \alpha_{s,T} +d)  ds.
\end{align*} Assume now that \eqref{eq:deriv-M-induction} holds true for some $h\ge1$. Then
\begin{multline*}
    \bar{M}_{F_T}^{(h+1)}(\theta) = f(T) \bar{M}_{F_T}^{(h)}(\theta)\\
    + \lambda\sum_{i=1}^h \binom{h-1}{i-1} \left[\bar{M}_{F_T}^{(h+1-i)}(\theta) \int_0^T \kappa^{(i)} \left(  \theta(c \alpha_{s,T}  +d) + \ell\rho \right) \left( c \alpha_{s,T} + d\right)^i ds\right. \\
    + \left.\bar{M}_{F_T}^{(h-i)}(\theta) \int_0^T \kappa^{(i+1)} \left(  \theta \left(c \alpha_{s,T} +d\right) + \ell\rho \right) \left( c \alpha_{s,T}  +d \right)^{i+1} ds\right].
\end{multline*}
The desired result
\begin{multline*}
    \bar{M}_{F_T}^{(h+1)}(\theta) = f(T) \bar{M}_{F_T}^{(h)}(\theta)\\
    + \lambda\sum_{i=1}^{h+1} \binom{h}{i-1} \bar{M}_{F_T}^{(h+1-i)}(\theta) \int_0^T \kappa^{(i+1)} \left(  \theta \left(c \alpha_{s,T}  +d\right) + \ell\rho \right) \left( c \alpha_{s,T}  + d\right)^{i+1} ds
\end{multline*}
follows after grouping like terms together and making use of the properties of the binomial function. This concludes the inductive step.

Equations \eqref{eq:H-ito-Mderiv} and \eqref{eq:deriv-M-induction} give equation \eqref{eq:LemmaH} for $h\geq 1$, as claimed.
\end{proof}

Using Lemma \ref{theo:MomentsR}, we can compute the moments in equation \eqref{Moment1}. First, we need to apply the binomial theorem to obtain
\begin{multline} \label{eq:binom-theorem}
 (P_T-1)^{n-k} \left(I_T- E_{\mathds{Q}}[I_T] \right)^k 
 = \sum_{\ell=0}^{n-k} \binom{n-k}{\ell} (-1)^{n-k-\ell} P_T^\ell \left(I_T - E_{\mathds{Q}}\left[ I_T \right]\right)^{k}.
\end{multline}
Thus, the approximation reduces to computing mixed moments of the form $E_{\mathds{Q}}\left[P_T^\ell \left(I_T - E_{\mathds{Q}}\left[ I_T \right]\right)^{k} \right]$, where $\ell$ and $k$ are non-negative integers. Using equations \eqref{eq:I-SDE} and \eqref{eq:Exp_IT} we have that
\begin{align}
I_T - E_{\mathds{Q}}\left[ I_T \right]  =  -  \lambda \kappa'(0) \int_0^T   \alpha_{s,T} ds  +  \int_0^T \alpha_{s,T} dZ_{\lambda s}. \label{eq:CProcess} 
\end{align} Lemma \ref{theo:MomentsR} then gives the following.

\begin{corollary} \label{cor:Moments}
For any $\ell\in\mathds{N}_0$ such that $\ell\rho < \hat{\kappa}$, and any $k\in\mathds{N}_0$ such that $\kappa$ is $k$ times continuously differentiable in an open interval containing $0$, we have
 \[
  E_{\mathds{Q}}\left[P_T^\ell \left( I_T - E_{\mathds{Q}}\left[ I_T \right] \right)^k \right] = e^{\lambda T (\kappa(\ell \rho)-\ell \kappa(\rho))} H_{\ell,k},
 \]
 where $H_{\ell,k}$ satisfies the recursive relationship
\begin{align}
  H_{\ell,0} &= 1, \nonumber\\
  H_{\ell,h} &=  \left(-  \lambda  \kappa'(0) \int_0^T \alpha_{s,T} ds \right)  H_{\ell,h-1} \nonumber\\
  & \quad +   \lambda \sum_{i=1}^h \binom{h-1}{i-1} H_{\ell,h-i} 
\kappa^{(i)} (\ell \rho) \int_0^T    \left( \alpha_{s,T} \right)^i ds \nonumber\\ 
 &= \kappa'(0) \left( \alpha_{0,T} -T\right) H_{\ell,h-1}\nonumber\\
  & \quad +   \sum_{i=1}^h \frac{1}{\lambda^{i-1}} \binom{h-1}{i-1} H_{\ell,h-i} 
\kappa^{(i)} (\ell \rho) \left[ T +  \sum_{j=1}^i  \frac{1}{j}\binom{i}{j} (-1)^j \alpha_{0,jT}\right] \label{eq:HMoments}
  \end{align}
 for all $h=1,\ldots,k$. 
 
\end{corollary}
\begin{proof}
The result comes from the application of Lemma \ref{theo:MomentsR} with $d=0$, $c=1$, \begin{align*}
f(t) & = - \lambda  \kappa'(0) \int_0^t  \alpha_{s,t} ds\\
& =  \kappa'(0) \left( \alpha_{0,t} -t\right),
\end{align*} and from the fact that for $i\in\mathds{N}$,
\begin{align*}
     \int_0^T    \left( \alpha_{s,T}\right)^i ds 
     & = \frac{1}{\lambda^i} \sum_{j=0}^i  \binom{i}{j} (-1)^j  e^{-j\lambda T} \int_0^T    e^{j \lambda s} ds \\
     & = 
     \frac{1}{\lambda^{i+1}}\left[ \lambda T +  \sum_{j=1}^i  \frac{1}{j}\binom{i}{j} (-1)^j \left( 1 - e^{-j \lambda T} \right)\right]\\
     & = \frac{1}{\lambda^i} \left[ T +  \sum_{j=1}^i  \frac{1}{j}\binom{i}{j} (-1)^j \alpha_{0,jT}\right]. 
     \end{align*}
\end{proof}

In conclusion, the approximation relies on the cumulant function $\kappa$ of $Z_1$ and its derivatives. These are known explicitly for many D--Ornstein-Uhlenbeck processes; see Remark \ref{remark:kappa}.

\begin{remark}\label{remark:kappa}
    Notice that to be able to use the formula of Lemma \ref{theo:MomentsR} we need the cumulant function of $Z_1$ and its derivatives. Depending on the type of D--Ornstein-Uhlenbeck process we have different formulas for the function $\kappa$. For example:
\begin{enumerate}
    \item If the squared volatility is an inverse Gaussian Ornstein-Uhlenbeck process \cite[p.~449]{nicolato2003option}, we have that the function $\kappa$ satisfies \begin{align}\label{eq:kappaIG}
        \kappa(\theta) &= \frac{a\theta}{\sqrt{b^2-2\theta}} \text{ for } a,b>0.
    \end{align} In this case  $\hat{\kappa} = \frac{b^2}{2}$.
    For $n\in \mathds{N}$, the $n^{\text{th}}$ derivative of $\kappa$ can be written as:
    \begin{align}\label{eq:kappa_der_IG}
        \kappa^{(n)} (\theta) & = \varphi_n a (b^2 - 2\theta)^{- \frac{2n-1}{2}}+ (2n-1)!! a \theta (b^2 -2 \theta)^{- \frac{2n+1}{2}},
    \end{align} where $\varphi_n$ is defined recursively as
    \begin{align}\label{eq:varphi}
        \varphi_n &=  \varphi_{n-1} (2n-3) + (2n-3)!!
    \end{align} with $\varphi_1 =1$. Observe that the function $\kappa$ and its derivatives are well defined when $\theta < \frac{b^2}{2}$.
    
    \item If the squared volatility follows a gamma Ornstein-Uhlenbeck process\break \cite[p.~449]{nicolato2003option}, the cumulant function is 
    \begin{align}\label{eq:kappaGamma}
       \kappa(\theta) & = \frac{a\theta}{b-\theta} \text{ for } a,b>0. 
    \end{align} For this function of $\kappa$ we have that $\hat{\kappa}=b$.
     In this case the $n^{\text{th}}$ derivative of the cumulant function $\kappa$ can be written as:
     \begin{align}\label{eq:kappa_der_gamma}
     \kappa^{(n)}(\theta) & =  n! a(b-\theta)^{-n} + n! a\theta (b-\theta)^{-n-1}.    
     \end{align}
      Notice that the function $\kappa$ and its derivatives are well defined when  $\theta <b$. 
\end{enumerate}
\end{remark}

\section{Approximation error}\label{ErrorSect}
The focus in this section is on the remainder term $R_N$ of the $N^{\text{th}}$ order Taylor expansion, given in \eqref{eq:Taylor-remainder-N}. The expectation in $R_N$ combines a mixed moment with an integral of a partial derivative of the Black-Scholes formula for the price of a put (see \eqref{eq:BSPut}--\eqref{eq:d+-}). We will start by observing that there is a pattern in the partial derivatives of $\BS_{\Put}$. This pattern can already be observed in the third order derivatives, which are
\begin{align}
    \frac{\partial^3 \BS_{\Put} }{\partial x^3} (x,y) &= \frac{-\phi(d_+)}{x^2y} ( d_+ +\sqrt{y}), \label{Derxxx}\\
    \frac{\partial^3 \BS_{\Put} }{\partial y^3} (x,y) &= \frac{x \phi(d_+)}{8 y^{\nicefrac{5}{2}}}  \left( \left(d_-d_+ -2\right)^2 -d_+^2 -d_-^2 -1 \right), \label{Deryyy}\\
    \frac{\partial^2 \BS_{\Put} }{\partial x^2 \partial y } (x,y) &=  \frac{\phi(d_+) }{2x y^{\nicefrac{3}{2}}} (d_-d_+ - 1), \label{Derxxy}\\
    \frac{\partial^2 \BS_{\Put} }{\partial x \partial y^2 } (x,y) &= \frac{-\phi(d_+)}{2y^2} \left(\frac{d_-d_+}{2}  - \frac{d_+}{2} - d_- \right), \label{Derxyy}
\end{align} 
for all $x,y\in\mathds{R}$ \citep{das2022closed}, where
\[\phi(z) = \frac{1}{\sqrt{2 \pi}} e^{-\frac{1}{2}z^2} \text{ for all } z \in \mathds{R}\]
is the density function of the standard normal distribution. On equations \eqref{Derxxx}--\eqref{Derxyy}, we have suppressed the argument $(x,y)$ from $d_{\pm}$. The pattern for higher order derivatives is as follows. 

\begin{proposition}\label{PropDer}
    Every partial derivative of $\BS_{\Put}$ of third or higher order takes the form
    \begin{equation} \label{eq:deriv-shape}
    \frac{\partial^{\lvert\xi\rvert} \BS_{\Put} }{\partial x^{\xi_x} \partial y^{\xi_y}} (x,y) = \frac{A_\xi}{x^n y^{\nicefrac{m}{2}}} \phi(d_+(x,y)) F_\xi (d_+(x,y), d_-(x,y), \sqrt{y}),
    \end{equation}
    where $\xi=(\xi_x,\xi_y)$, $\xi_x,\xi_y \in\mathds{N}_0$, $m\in\mathds{N}$, $n\in\mathds{Z}$, $\lvert\xi\rvert = \xi_x+\xi_y\ge3$, $A_\xi\in\mathds{R}$ and $F_\xi$ is a polynomial in $d_+(x,y)$, $d_-(x,y)$ and $\sqrt{y}$, with degree at least one in either $d_+(x,y)$ or $d_-(x,y)$.
\end{proposition}

\begin{proof}
    For brevity we suppress the arguments $(x,y)$ in this proof. The claim holds when $\lvert\xi\rvert=3$ by \eqref{Derxxx}--\eqref{Derxyy}; we proceed by induction. Suppose that \eqref{eq:deriv-shape} holds for some $\xi$ with $\lvert\xi\rvert\ge3$.
    Noting that
    \begin{align*}
        \frac{\partial d_\pm}{\partial x} &= \frac{1}{x\sqrt{y}}, & \frac{\partial d_\pm}{\partial y} &= -\frac{d_\mp}{2y}, & \frac{\partial \sqrt{y}}{\partial y} &= -\frac{\sqrt{y}}{2y},
    \end{align*}
    we obtain
    \begin{align*}
        \frac{\partial \phi}{\partial x}(d_+) &= - \frac{1}{x\sqrt{y}} \phi(d_+) d_+, &
        \frac{\partial \phi}{\partial y}(d_+) &= \frac{1}{2y}\phi(d_+)d_+d_-.
    \end{align*}
    Observe furthermore that
    \begin{align*}
        \frac{\partial{F}_\xi}{\partial x}(d_+, d_-, \sqrt{y}) &= \frac{1}{x\sqrt{y}} G_\xi (d_+, d_-, \sqrt{y}), \\
        \frac{\partial{F}_\xi}{\partial y}(d_+, d_-, \sqrt{y}) &= \frac{1}{2y} H_\xi (d_+, d_-, \sqrt{y})
    \end{align*}
    where $G_\xi,H_\xi$ are polynomials in $d_+$, $d_-$ and $\sqrt{y}$. 

    We now differentiate $\frac{\partial^{\lvert\xi\rvert} \BS_{\Put} }{\partial x^{\xi_x} \partial y^{\xi_y}}$ with respect to $x$ and $y$, respectively. To this end,
     \begin{multline*}
        \frac{\partial^{\lvert\xi\rvert+1} \BS_{\Put} }{\partial x^{\xi_x+1} \partial y^{\xi_y}} \\
        \begin{aligned}
        &= \frac{\partial}{\partial x} \frac{\partial^{\lvert\xi\rvert} \BS_{\Put} }{\partial x^{\xi_x} \partial y^{\xi_y}} \\
        &= \frac{A_\xi}{x^{n+1} y^{\nicefrac{(m+1)}{2}}}\phi(d_+) \left[-\left(d_+ + n\sqrt{y}\right) F_\xi (d_+, d_-, \sqrt{y}) + G_\xi (d_+, d_-, \sqrt{y})\right]
        \end{aligned}
    \end{multline*}   
    and
     \begin{multline*}
        \frac{\partial^{\lvert\xi\rvert+1} \BS_{\Put} }{\partial x^{\xi_x} \partial y^{\xi_y+1}} \\
        \begin{aligned}
        &= \frac{\partial}{\partial y} \frac{\partial^{\lvert\xi\rvert} \BS_{\Put} }{\partial x^{\xi_x} \partial y^{\xi_y}} \\
        &= \frac{A_\xi}{2x^n y^{\nicefrac{m}{2}+1}} \phi(d_+) \left[\left(d_+d_--m\right) F_\xi (d_+, d_-, \sqrt{y}) 
        + H_\xi (d_+, d_-, \sqrt{y})\right].
        \end{aligned}
    \end{multline*}  
    It follows that both $\frac{\partial^{\lvert\xi\rvert+1} \BS_{\Put} }{\partial x^{\xi_x+1} \partial y^{\xi_y}}$ and $\frac{\partial^{\lvert\xi\rvert+1} \BS_{\Put} }{\partial x^{\xi_x} \partial y^{\xi_y+1}}$ is of the form claimed, which completes the inductive step.
\end{proof}

It turns out that the partial derivatives that appear in the Taylor remainder term \eqref{eq:Taylor-remainder-N} are bounded, in the following sense.

\begin{proposition} \label{prop:diff-bound}
    For every $\xi=(\xi_x,\xi_y)$, $\xi_x,\xi_y\in\mathds{N}_0$ such that $\lvert\xi\rvert = \xi_x+\xi_y\ge3$, there exists a function $M_\xi:[0,\infty)\times(0,\infty)\rightarrow[0,\infty)$ such that
    \[
    \sup_{u\in(0,1)}\left\lvert \frac{\partial^{\lvert\xi\rvert} \BS_{\Put} }{\partial x^{\xi_x} \partial y^{\xi_y}} ((1-u)S_0+uP_TS_0,(1-u)E_{\mathds{Q}}[I_T]+uI_T) \right\rvert \le M_\xi(T,K)
    \]
    almost surely. Furthermore, the function $M_\xi$ is bounded and it satisfies
    \begin{align}
    \lim_{T\rightarrow0}M_\xi(T,K)  & = \infty  \text{ for all } K>0,\label{eq:MT0} \\
        \lim_{K\rightarrow0}M_\xi(T,K) & = \lim_{K\rightarrow\infty}M_\xi(T,K) = 0 \text{ for all } T>0, \label{eq:MK0Inf}\\
        \text{if } r\neq 0, \text{ then } & \lim_{T\rightarrow\infty}M_\xi(T,K) = 0 \text{ for all } K>0.\label{eq:MTInf}
    \end{align}
\end{proposition}

\begin{proof}
   See Proposition \ref{cor:boundsDer} in Appendix \ref{DerBoundsect}.
\end{proof}

Based on the result given in Proposition \ref{prop:diff-bound}, it is possible to construct a bound for the error term and show that the error disappears for small and large values of $K$. 

\begin{theorem} \label{Theo:Error}
    For any $N\in\mathds{N}$, if 
\[
     E_{\mathds{Q}}\left[\lvert P_T-1\rvert^{n-k} \left\lvert I_T- E_{\mathds{Q}}[I_T] \right\rvert^k \right] < \infty
\] for all $n=1,\ldots,N+1$ and $k=0,\ldots,n$, then the remainder term of the Taylor approximation of Theorem \ref{ApproxThm} satisfies
\begin{align}
     \lvert R_N\rvert & \le \frac{1}{(N+1)!} \sum_{n=0}^{N+1} \binom{N+1}{n}S_0^{N+1-n} \nonumber\\
     & \quad \times E_{\mathds{Q}}\left[ \left(P_T-1\right)^{N+1-n} \left( I_T -E_{\mathds{Q}}[I_T] \right)^n\right]  M_{(N+1-n,n)}(T,K)\label{eq:Error_simplified}
\end{align}
where $M_\xi(T,K)$ is as in Proposition \ref{prop:diff-bound}. In addition,
\begin{align}
        &\lim_{K\rightarrow0} |R_N|  = \lim_{K\rightarrow\infty} |R_N| = 0 \text{ for all } T>0. \label{eq:limitK0Inf}
        \end{align}
\end{theorem}

\begin{proof} 
Direct application of Theorem \ref{ApproxThm} and  Proposition \ref{prop:diff-bound} gives \eqref{eq:Error_simplified}. Equation \eqref{eq:limitK0Inf} follows from \eqref{eq:MK0Inf}.
\end{proof}

\subsection{Error bound when $\rho\leq 0$}\label{Sect:ErrorBoundrholeq0}
In this section, we derive a more precise  error bound when $\rho\leq0$. From Theorem \ref{Theo:Error} and the Cauchy-Schwarz inequality, we can obtain the following upper bound for the error\begin{align} 
 |R_N|
 \leq \frac{1}{(N+1)!}\sum_{n=0}^{N+1} \binom{N+1}{n}S_0^{N+1-n}M_{(N+1-n,n)} (T,K) \nonumber\\
 \times E_{\mathds{Q}}\left[(P_T-1)^{2N+2-2n}\right]^{\nicefrac{1}{2}} E_{\mathds{Q}}\left[\left(I_T - E_{\mathds{Q}}[I_T]\right)^{2n} \right]^{\nicefrac{1}{2}}\label{eq:Taylor-remainder-N2}
\end{align} for $ N \in \mathds{N} \setminus \{1\}$. Since $\rho \leq 0$, the moments that appear in expression \eqref{eq:Taylor-remainder-N2} are finite as long as $\kappa$ is $2N+2$ times continuously differentiable in an open interval containing $0$ (see Corollary \ref{cor:Moments}). 

In the case $\rho=0$, Theorem \ref{Theo:Error} gives the error bound
\begin{align}\label{eq:RNBoundRho0}
    | R_N |
 & \leq \frac{1}{(N+1)!}\left|E_{\mathds{Q}}\left[\left(I_T - E_{\mathds{Q}}[I_T]\right)^{N+1}\right]\right| M_{(0,N+1)} (T,K). 
\end{align}

From equation \eqref{eq:Taylor-remainder-N2}, we observe that the error depends on the moments of   $P_T-1$ and $I_T - E_{\mathds{Q}}[I_T]$. The  moments of  $I_T - E_{\mathds{Q}}[I_T]$ were studied in Corollary \ref{cor:Moments}. Now, we need to analyze the moments of $P_T-1$.

\begin{proposition}\label{Prop:GN}When $\rho\leq 0$ and for any $N\in\mathds{N}_0$ such that $\kappa$ is $N$ times continuously differentiable in an open interval containing $0$, we have
\begin{align*}
 \left|E_{\mathds{Q}}\left[(P_T-1)^{N}\right]\right|    &\leq  G_N,
\end{align*}  where $G_{N}$ is defined as
\begin{equation}
    G_N = E_{\mathds{Q}}\left[\left(-\lambda T\kappa(\rho) \left( \frac{1}{2} e^{-\lambda T\kappa (\rho)} +\frac{1}{2}\right) + |\rho| Z_{\lambda T} \right)^{N}\right], \label{eq:GN}
\end{equation} and it satisfies the recursive relationship
 \begin{equation}\label{eq:recusiveMoment3}\left.
  \begin{aligned}
  G_{0} &= 1, \\
  G_{h} &= -\lambda T \kappa(\rho)  \left(  \frac{1}{2} e^{-\lambda T\kappa (\rho)} +\frac{1}{2} \right) G_{h-1}  + \sum_{i=1}^h \binom{h-1}{i-1} G_{h-i} 
\kappa^{(i)} (0) \lambda T |\rho|^i
  \end{aligned}\right\}
 \end{equation}
 for all $h=1,\ldots,N$. 
\end{proposition}
\begin{proof}
    Since $\rho\leq 0$, $Z$ is non-negative and  $\lambda T \kappa(\rho)\leq 0$, we obtain
    \begin{align*}
        P_T - 1 & \leq e^{-\lambda T\kappa(\rho)}-1\\
        & \leq - \frac{1}{2} \lambda T\kappa(\rho)\left(  e^{-\lambda T\kappa(\rho)} + 1\right),
    \end{align*}where the last equality comes from the fact that $e^x - 1 \leq \frac{x}{2}\left(e^x + 1 \right)$ for all $x\geq 0$. From $\ln(x) \leq x-1 $ for all $x\geq 0$, we get the upper bound
    \begin{align*}
        1-P_T &\leq -\ln(P_T)\\
        & = - \rho Z_{\lambda T} + \lambda T \kappa(\rho)\\
        &\leq - \rho Z_{\lambda T} = |\rho| Z_{\lambda T} .
    \end{align*} It follows that 
    \begin{align*}
        |P_T-1|&\leq  - \frac{1}{2} \lambda T\kappa(\rho)\left(  e^{-\lambda T\kappa(\rho)} + 1\right)  + |\rho| Z_{\lambda T}.
    \end{align*} Hence we have
    \begin{align}\label{eq:BoundMomentP}
         E_{\mathds{Q}}\left[\left|P_T-1\right|^{N}\right] &\leq  E_{\mathds{Q}}\left[\left(- \frac{1}{2} \lambda T\kappa(\rho) \left(  e^{-\lambda T\kappa (\rho)} +1\right) + |\rho| Z_{\lambda T} \right)^{N}\right]
    \end{align} for all $N\in\mathds{N}_0$. To obtain the recursive relation in equation \eqref{eq:recusiveMoment3}, we just need to apply Lemma \ref{theo:MomentsR} with $c=0$, $d= |\rho|$ and \[ 
    f(t) = -  \frac{1}{2} \lambda t\kappa(\rho) \left( e^{-\lambda t\kappa (\rho)} + 1\right).
    \]
\end{proof}
As one can imagine, the bound for the error will depend on the type of Ornstein-Uhlenbeck process that the squared volatility satisfies. We differentiate between the cases when the squared volatility follows a gamma Ornstein-Uhlenbeck process or an inverse Gaussian Ornstein-Uhlenbeck process.

\subsubsection{Gamma Ornstein-Uhlenbeck process}
If we assume that the variance process follows a gamma Ornstein-Uhlenbeck process with parameters $a,b>0$, one can construct an error bound of the form $\mathcal{O}\left(\frac{1}{b^{N+1}}\right)$. But first, we need an upper bound for the moments of $I_t  -E_\mathds{Q}\left[I_T\right]$ and $P_T-1$. 
\begin{theorem} \label{TheoremMomentsGamma} When the variance is a gamma Ornstein-Uhlenbeck process, for $N\in \mathds{N}$ we have
\begin{equation}
  \left|E_{\mathds{Q}}\left[ \left(I_T -  E_{\mathds{Q}}\left[ I_T \right]\right)^N \right]\right| \leq  \frac{aT}{\lambda b^N} f_N\left(a,\frac{1}{\lambda},T\right), \label{eq:ineGamma}    
\end{equation} where $f_{N}$ is a polynomial of $a$, $\frac{1}{\lambda}$ and $T$ with non-negative coefficients. The bound for $G_N$ defined in Proposition \ref{Prop:GN}, satisfies
\begin{equation}
      G_N  \leq   \frac{a\lambda T |\rho|}{b^N} g_N\left(a,\lambda,T,|\rho|,e^{-\lambda T \kappa(\rho)}  \right) \label{eq:ineGamma2}
\end{equation} for $N\in \mathds{N}$, where $g_{N}$ is a polynomial in $a$, $\lambda$, $T$, $|\rho|$ and $e^{-\lambda T \kappa(\rho)}$ with non-negative coefficients.
\end{theorem} 
\begin{proof}
From \eqref{eq:kappa_der_gamma}, the derivatives of $\kappa$ satisfy
\begin{align}
    \kappa^{(i)}(0) & = \frac{i! a}{ b^i} \label{eq:kappa0Gamma}
\end{align} for all  $i \in \mathds{N}$. We prove inequality \eqref{eq:ineGamma} by induction. From Corollary \ref{cor:Moments}, observe that for $N=1$ and $N=2$ we have\begin{align*}
\left|E_{\mathds{Q}}\left[ \left(I_T -  E_{\mathds{Q}}\left[ I_T \right]\right) \right]\right| & = 0,\\
\left|E_{\mathds{Q}}\left[ \left(I_T -  E_{\mathds{Q}}\left[ I_T \right]\right)^2 \right]\right| & = \frac{2a}{ \lambda b^2} \int_0^T\left(1- e^{-\lambda (T-s)}\right)^2ds \leq \frac{2aT}{\lambda b^2},
\end{align*} so the result in $\eqref{eq:ineGamma}$ is satisfied for $N=1$ and  $N=2$ with $f_1 =0$ and  $f_2 =2$, 

Let us assume that the claim in $\eqref{eq:ineGamma}$ is true for $N$ and we will prove it for the case $N+1$. From inequality $\alpha_{s,T} \leq \frac{1}{\lambda}$, equation \eqref{eq:HMoments} in Corollary \ref{cor:Moments} and equation \eqref{eq:kappa0Gamma} we have\begin{multline*}
    \left|E_{\mathds{Q}}\left[ \left(I_T -  E_{\mathds{Q}}\left[ I_T \right]\right)^{N+1} \right]\right|  \leq \frac{a}{b}  \left|E_{\mathds{Q}}\left[ \left(I_T -  E_{\mathds{Q}}\left[ I_T \right]\right)^{N} \right]\right|T \\
    \begin{aligned}
    & \quad + \sum_{i=1}^{N+1} \frac{1}{\lambda^{i-1}}\binom{N}{i-1} \frac{i! a}{b^i} \left|E_{\mathds{Q}}\left[ \left(I_T -  E_{\mathds{Q}}\left[ I_T \right]\right)^{N+1-i} \right]\right| T \\
    & \leq  \frac{a}{b}  \frac{aT}{\lambda b^{N}} f_{N} \left(a,\frac{1}{\lambda},T\right) T  \\
    & \quad + \sum_{i=1}^{N-1} \frac{1}{\lambda^{i-1}} \frac{N!}{(N+1-i)!}  \frac{i a}{b^i} \frac{aT}{\lambda b^{N+1-i}}  f_{N+1-i}\left(a,\frac{1}{\lambda},T\right) T \\
    & \quad + \frac{1}{\lambda^{N}} \frac{(N+1)! a}{b^{N+1}} T\\
    & = \frac{aT}{\lambda b^{N+1}} f_{N+1}\left(a,\frac{1}{\lambda},T\right),
  \end{aligned}
\end{multline*}where 
\begin{align*}
    f_{N+1}\left(a,\frac{1}{\lambda},T\right) & = a f_{N} \left(a,\frac{1}{\lambda},T\right)T + \sum_{i=1}^{N-1} \frac{ia}{\lambda^{i-1}} \frac{N!}{(N+1-i)!}  f_{N+1-i}\left(a,\frac{1}{\lambda},T\right)T\\
    & \quad + \frac{(N+1)!}{\lambda^{N-1}}.
\end{align*} We have just shown that the claim \eqref{eq:ineGamma} holds for any $N \in \mathds{N}$.

We will now prove inequality \eqref{eq:ineGamma2} by induction. From \eqref{eq:kappaGamma} and \eqref{eq:GN}, we observe that 
    \begin{align*}
        G_1 & = -\lambda T \frac{a \rho}{b-\rho}  \left( \frac{1}{2} e^{-\lambda T\kappa (\rho)} +\frac{1}{2}\right)+ \frac{a|\rho|}{b}\lambda T\\
        &\leq \frac{a|\rho| \lambda T}{b}  \left(  \frac{1}{2}e^{-\lambda T\kappa (\rho)} +\frac{3}{2} \right),
    \end{align*} where we have used the fact that $\rho\leq 0$ in the last inequality, and hence 
    \[g_1 \left(a,\lambda,T,|\rho|,e^{-\lambda T \kappa(\rho)}  \right) =   \frac{1}{2}e^{-\lambda T\kappa (\rho)} +\frac{3}{2}.\] So the inequality \eqref{eq:ineGamma2} is satisfied for $N=1$.
    
    If we assume that the result in \eqref{eq:ineGamma2} is true for $N$, from \eqref{eq:kappaGamma} and \eqref{eq:GN} for the case $N+1$ we obtain 
      \begin{align*}
        G_{N+1} & = -\lambda T \frac{a \rho}{b-\rho}  \left( \frac{1}{2} e^{-\lambda T\kappa (\rho)} +\frac{1}{2}\right) G_{N}  +  \sum_{i=1}^{N+1} \binom{N}{i-1} G_{N+1-i} \frac{i! a }{b^i} \lambda T |\rho|^i \\
        & \leq -\lambda T \frac{a \rho}{b}  \left( \frac{1}{2} e^{-\lambda T\kappa (\rho)} +\frac{1}{2}\right) \frac{a|\rho| \lambda T}{b^{N}} g_{N}\left(a,\lambda,T,|\rho|,e^{-\lambda T \kappa(\rho)}  \right) \\
        &\quad +  \sum_{i=1}^{N} \frac{N!}{(N+1-i)!} \frac{a|\rho| \lambda T}{b^{N+1-i}} g_{N+1-i}\left(a,\lambda,T,|\rho|,e^{-\lambda T \kappa(\rho)}  \right) \frac{i a}{b^i} \lambda T |\rho|^i \\
        & \quad + \frac{(N+1)! a}{b^{N+1}}\lambda T |\rho|^{N+1} \\
        & = \frac{a\lambda T |\rho|}{b^{N+1}} g_{N+1}\left(a,\lambda,T,|\rho|,e^{-\lambda T \kappa(\rho)}  \right),
    \end{align*} where
    \begin{multline*} 
    g_{N+1}\left(a,\lambda,T,|\rho|,e^{-\lambda T \kappa(\rho)}  \right)\\
    \begin{aligned}
         & =   \left( \frac{1}{2} e^{-\lambda T\kappa (\rho)} +\frac{1}{2}\right) a|\rho| \lambda T g_{N}\left(a,\lambda,T,|\rho|,e^{-\lambda T \kappa(\rho)}  \right)\\
        &\quad + \sum_{i=1}^{N} \frac{N!}{(N+1-i)!} g_{N+1-i}\left(a,\lambda,T,|\rho|,e^{-\lambda T \kappa(\rho)}  \right) i a \lambda T |\rho|^i + (N+1)! |\rho|^{N}.
    \end{aligned}
    \end{multline*} We have just proved \eqref{eq:ineGamma2}. 
\end{proof}
Combining the results from Theorem \ref{TheoremMomentsGamma} and Proposition \ref{prop:diff-bound} with inequality \eqref{eq:Taylor-remainder-N2}, we arrive at the following error bound.
\begin{corollary}\label{cor:ErrorBoundGamma}
When the squared volatility follows a gamma Ornstein-Uhlenbeck process and $\rho\leq 0$,  the remainder term of the $N^{\text{th}}$ order Taylor approximation of Theorem \ref{ApproxThm} is bounded by
\begin{align*}
     \lvert R_N\rvert &    \leq  \frac{\sqrt{aT}}{ b^{N+1} (N+1)! }\left(\sum_{n=1}^{N} \left( \binom{N+1}{n}S_0^{N+1-n}M_{(N+1-n,n)} (T,K) \right. \right. \nonumber\\
 & \quad \times \left.\sqrt{ aT |\rho|  g_{2N+2-2n} \left(a,\lambda,T,|\rho|,e^{-\lambda T \kappa(\rho)}  \right) f_{2n} \left(a,\frac{1}{\lambda},T \right) } \right) \nonumber\\
 & \quad  +  S_0^{N+1} M_{(N+1,0)}(T,K)\sqrt{ \lambda |\rho|  g_{2N+2} \left(a,\lambda,T,|\rho|,e^{-\lambda T \kappa(\rho)}  \right) }  \\ 
 & \quad \left. + M_{(0,N+1)}(T,K)\sqrt{ \frac{1}{\lambda}  f_{2N+2} \left(a,\frac{1}{\lambda},T \right) }  \right)
\end{align*} for $N\in \mathds{N}\setminus \{1\}$, where  $M_{(N+1-n,n)}$, $ f_{2n}$ and $g_{2N+2-2n}$ satisfy the properties in Proposition \ref{prop:diff-bound} and Theorem \ref{TheoremMomentsGamma}. In particular when $\rho=0$, the error term can be bounded by
\begin{align}\label{eq:BoundGamma}
     \lvert R_N\rvert & \le  \frac{1}{(N+1)!} \frac{aT}{\lambda b^{N+1}} f_{N+1}\left(a,\frac{1}{\lambda},T\right)M_{(0,N+1)}(T,K).
    \end{align}
\end{corollary}

Since $\kappa(\rho) = \frac{a\rho}{b-\rho}$ is a decreasing function with respect to the parameter $b$, Corollary \ref{cor:ErrorBoundGamma} tells us that if we fix all parameters in our model and we only allow to vary the parameter $b$, then the error becomes smaller as $b$ becomes bigger. In this case, the error bound of an option with strike $K$ and expiration date $T$ have the form $\mathcal{O}\left( \frac{1}{b^{N+1}}\right)$.

\subsubsection{Inverse Gaussian Ornstein-Uhlenbeck process}

If the variance process follows an inverse Gaussian Ornstein-Uhlenbeck process with parameters $a,b>0$, we can give an error bound for the $N^{\text{th}}$ order approximation of the form $\mathcal{O}\left(\frac{1}{b^{N+1}}\right)$. In the case $\rho=0$, the error bound has the form $\mathcal{O}\left(\frac{1}{b^{N+2}}\right)$. As we did before, we need to bound the moments of $I_T  -E_\mathds{Q}\left[I_T\right]$ and $P_T-1$.

\begin{theorem} \label{TheoremMomentsIG} When the squared volatility follows an inverse Gaussian Ornstein-Uhlenbeck process, we have
\begin{equation}
\left|E_{\mathds{Q}}\left[ \left(I_T -  E_{\mathds{Q}}\left[ I_T \right]\right)^N \right]\right| \leq \frac{aT}{\lambda b^{N+1}} f_N\left(a,\frac{1}{\lambda},\frac{1}{b},T\right)\label{eq:ineIG}    
\end{equation} for $N\in \mathds{N}$, where $f_{N}$ is a polynomial of $a$, $\frac{1}{\lambda}$, $\frac{1}{b}$  and $T$ with non-negative coefficients.  The bound $G_N$ of Proposition \ref{Prop:GN} satisfies \begin{equation}
   G_N \leq   \frac{a\lambda T |\rho|}{b^N} g_N\left(a,\lambda,\frac{1}{b},T,|\rho|,e^{-\lambda T \kappa(\rho)}  \right) \label{eq:ineIG2} 
\end{equation} for $N\in \mathds{N}$, where $g_{N}$ is a polynomial in $a$, $\lambda$, $\frac{1}{b}$, $T$, $|\rho|$ and $e^{-\lambda T \kappa(\rho)}$ with non-negative coefficients.
\end{theorem}
\begin{proof}
  From equation \eqref{eq:kappa_der_IG} we obtain that
\begin{align}
    \kappa^{(i)}(0) & = \frac{\varphi_i a}{ b^{2i-1}} \text{ for all } i \in \mathds{N}, \label{eq:kappa0IG}
\end{align} where $\varphi_i$ is defined recursively as $\varphi_1 =1$ and
\begin{align*}
    \varphi_i & = \varphi_{i-1} (2n-3) + (2n-3)!! \text{ for all } i\in \mathds{N}\setminus\{1\}.
\end{align*} For $N=1$ and $N=2$ we have that
\begin{align*}
    \left|E_{\mathds{Q}}\left[ \left(I_T -  E_{\mathds{Q}}\left[ I_T \right]\right) \right]\right| & = 0\\
    \left|E_{\mathds{Q}}\left[ \left(I_T -  E_{\mathds{Q}}\left[ I_T \right]\right)^2 \right]\right| & =\frac{1}{\lambda} \frac{\varphi_2 a }{b^3} \int_0^T\left(1- e^{-\lambda (T-s)}\right)^2ds \leq  \frac{\varphi_2 a T}{\lambda b^3},
\end{align*} so the result in \eqref{eq:ineIG} is satisfied with $f_1=0$ and $f_2 =  \varphi_2$. If we assume that the claim \eqref{eq:ineIG} is true for $N$, then by application of \eqref{eq:HMoments} in Corollary \ref{cor:Moments} and equation \eqref{eq:kappa0IG} we have 
\begin{multline*}
    \left|E_{\mathds{Q}}\left[ \left(I_T -  E_{\mathds{Q}}\left[ I_T \right]\right)^{N+1} \right]\right| \\ 
    \begin{aligned}
    & \leq \frac{a}{b}  \left|E_{\mathds{Q}}\left[ \left(I_T -  E_{\mathds{Q}}\left[ I_T \right]\right)^{N} \right]\right| T \\
    & \quad + \sum_{i=1}^{N+1} \frac{1}{\lambda^{i-1}}\binom{N}{i-1} \frac{\varphi_i a}{b^{2i-1}} \left|E_{\mathds{Q}}\left[ \left(I_T -  E_{\mathds{Q}}\left[ I_T \right]\right)^{N+1-i} \right]\right| T \\
    & \leq  \frac{a}{b}  \frac{aT}{\lambda b^{N+1}} f_{N} \left(a,\frac{1}{\lambda},\frac{1}{b},T\right) T \\
    & \quad + \sum_{i=1}^{N-1} \frac{1}{\lambda^{i-1}} \binom{N}{i-1}  \frac{\varphi_i a}{b^{2i-1}} \frac{a T}{\lambda b^{N-i+2}}  f_{N+1-i}\left(a,\frac{1}{\lambda},\frac{1}{b},T\right)T\\
    & \quad + \frac{1}{\lambda^{N}} \frac{\varphi_{N+1} a T}{b^{2N+1}}  \\
    & = \frac{aT}{\lambda b^{N+2}} f_{N+1}\left(a,\frac{1}{\lambda},\frac{1}{b},T\right),
    \end{aligned}
\end{multline*}where 
\begin{align*}
    f_{N+1}\left(a,\frac{1}{\lambda},\frac{1}{b},T\right) & = a f_{N} \left(a,\frac{1}{\lambda},\frac{1}{b},T\right)T\\
    & \quad + \sum_{i=1}^{N-1} \frac{1}{\lambda^{i-1}} \binom{N}{i-1}  \frac{\varphi_i a}{b^{i-1}}  f_{N+1-i}\left(a,\frac{1}{\lambda},\frac{1}{b},T\right)T\\
    & \quad + \frac{1}{\lambda^{N-1}} \frac{\varphi_{N+1}}{b^{N-1}}.
\end{align*} We have just proved inequality \eqref{eq:ineIG} by induction.

The second inequality will be proved by induction as well. Using the derivatives of the cumulant function \eqref{eq:kappaIG} and \eqref{eq:GN} in Proposition \ref{Prop:GN}, for the initial case we have
    \begin{align*}
        G_1 & = -\lambda T \frac{a \rho}{\sqrt{b^2-2\rho}}  \left( \frac{1}{2} e^{-\lambda T\kappa (\rho)} +\frac{1}{2}\right)+ \frac{a|\rho|}{b}\lambda T\\
        &\leq \frac{a|\rho| \lambda T}{b}  \left(  \frac{1}{2}e^{-\lambda T\kappa (\rho)} +\frac{3}{2} \right),
    \end{align*} where we have used the fact that $\rho\leq 0$ and $g_1$ is \[g_1\left(a,\frac{1}{\lambda},\frac{1}{b},T\right) = \frac{1}{2}e^{-\lambda T\kappa (\rho)} +\frac{3}{2} .\]
    If we assume that the result in \eqref{eq:ineIG2} is true for $N$, for $N+1$ we obtain the following inequality
      \begin{align*}
        G_{N+1} & = -\lambda T \frac{a \rho}{\sqrt{b^2-2\rho}}  \left( \frac{1}{2} e^{-\lambda T\kappa (\rho)} +\frac{1}{2}\right) G_{N}\\
        & \quad +  \sum_{i=1}^{N+1} \binom{N}{i-1} G_{N+1-i} \frac{\varphi_i a}{b^{2i-1}} \lambda T |\rho|^i \\
        & \leq -\lambda T \frac{a \rho}{b}  \left( \frac{1}{2} e^{-\lambda T\kappa (\rho)} +\frac{1}{2}\right) \frac{a|\rho| \lambda T}{b^{N}} g_{N}\left(a,\lambda,T,\frac{1}{b},|\rho|,e^{-\lambda T \kappa(\rho)}  \right) \\
        &\quad +  \sum_{i=1}^{N} \binom{N}{i-1}  \frac{a|\rho| \lambda T}{b^{N+1-i}} g_{N+1-i}\left(a,\lambda,T,\frac{1}{b},|\rho|,e^{-\lambda T \kappa(\rho)}  \right) \frac{\varphi_i a }{b^{2i-1}} \lambda T |\rho|^i \\
        & \quad + \frac{\varphi_{N+1} a}{b^{2N+1}}\lambda T |\rho|^{N+1} \\
        & = \frac{a\lambda T |\rho|}{b^{N+1}} g_{N+1}\left(a,\lambda,T,\frac{1}{b},|\rho|,e^{-\lambda T \kappa(\rho)} \right),
    \end{align*} where
    \begin{multline*} 
       g_{N+1}\left(a,\lambda,T,\frac{1}{b},|\rho|,e^{-\lambda T \kappa(\rho)}  \right)\\
    \begin{aligned}
        & =   \left( \frac{1}{2} e^{-\lambda T\kappa (\rho)} +\frac{1}{2}\right) a|\rho| \lambda T g_{N}\left(a,\lambda,T,\frac{1}{b},|\rho|,e^{-\lambda T \kappa(\rho)}  \right)\\
        &\quad + \sum_{i=1}^{N} \binom{N}{i-1} g_{N+1-i}\left(a,\lambda,T,\frac{1}{b},|\rho|,e^{-\lambda T \kappa(\rho)}  \right) \frac{\varphi_i a \lambda T |\rho|^i}{b^{i-1}} \\
        & \quad+ \frac{\varphi_{N+1} |\rho|^{N}}{b^{N}}.
       \end{aligned}       
    \end{multline*} This ends the induction procedure.
\end{proof}

Finally, from Theorem \ref{TheoremMomentsIG}, Proposition \ref{Prop:GN} and equation \eqref{eq:Taylor-remainder-N2}, we have the following error estimate.
\begin{corollary}\label{cor:IGErrorBound}
When the variance  follows an inverse Gaussian Ornstein-Uhlenbeck process and $\rho\leq 0$,  the remainder term of the $N^{\text{th}}$ order Taylor approximation of Theorem \ref{ApproxThm} satisfies
    \begin{align*}
     \lvert R_N\rvert &    \leq  \frac{\sqrt{aT}}{ b^{N+1} (N+1)! }\left(\sum_{n=1}^{N} \left( \binom{N+1}{n}S_0^{N+1-n}M_{(N+1-n,n)} (T,K) \right. \right. \nonumber\\
 & \quad \times \left.\sqrt{ \frac{aT|\rho|}{b} g_{2N+2-2n} \left(a,\lambda,\frac{1}{b},T,|\rho|,e^{-\lambda T \kappa(\rho)}  \right)  f_{2n} \left(a,\frac{1}{\lambda},\frac{1}{b},T \right) } \right) \nonumber\\
 & \quad +  M_{(0,N+1)}(T,K) \sqrt{ \frac{1}{b^2 \lambda}  f_{2N+2} \left(a,\frac{1}{\lambda},\frac{1}{b},T \right) } \\
 & \left. \quad + S_0^{N+1} M_{(N+1,0)}(T,K)  \sqrt{\lambda |\rho|  g_{2N+2} \left(a,\lambda,\frac{1}{b},T,|\rho|,e^{-\lambda T \kappa(\rho)}  \right) } \right)
    \end{align*} for $N\in \mathds{N}\setminus \{1\}$, where  $M_{(N+1-n,n)}$, $f_{2n}$ and $g_{2N + 2 -2n}$ satisfy the properties in Proposition \ref{prop:diff-bound} and Theorem \ref{TheoremMomentsIG}. In particular when $\rho=0$, the error term is bounded by
\begin{align}\label{eq:BoundGammaError}
     \lvert R_N\rvert & \le  \frac{1}{(N+1)!} \frac{aT}{\lambda b^{N+2}} f_{N+1}\left(a,\frac{1}{\lambda},\frac{1}{b},T\right)M_{(0,N+1)}(T,K).
    \end{align}
\end{corollary}
In this case, the cumulant function satisfies $\kappa(\rho) = \frac{a\rho}{\sqrt{b^2-2\rho}}$ and it is a decreasing function with respect to the parameter $b$. From Corollary \ref{cor:IGErrorBound} the error is bounded by a function of the form  $\mathcal{O}\left( \frac{1}{b^{N+1}}\right)$ and when $\rho=0$, the error bound has the form $\mathcal{O}\left( \frac{1}{b^{N+2}}\right)$. 

\section{Numerical examples}\label{NumExamplesSect}

In this section, we compare the approximation result with the option price closed formula, computed using the characteristic function and the formula provided by \citet{carr1999option}, when the variance is an inverse Gaussian or a gamma Ornstein-Uhlenbeck process. Appendix \ref{CFsect} contains the characteristic functions that we use to compare our approximation results.

The homogeneity property of the Black-Scholes put price \citep{joshi2001log} gives us\[\Pi_P(S_0,K,T) = K\Pi_P\left(\frac{S_0}{K},1,T\right).\]For simplicity the numerical results in this section correspond to the put prices of the form  $\Pi_P\left(x,1,T\right)$ for various values of the moneyness $x$.

To compute the approximation formula in \eqref{eq:Taylor-approx-N}, we need the first $N$ derivatives of the cumulant generating function $\kappa$ and $2^{\text{nd}}$ to $N^{\text{th}}$ order derivatives of $\BS_{\Put}$. The derivatives for the cumulant generating function can be obtained from Remark \ref{remark:kappa} for inverse Gaussian and gamma Ornstein-Uhlenbeck processes. Higher order derivatives of $\BS_{\Put}$ can be derived from the second order derivatives shown in Remark \ref{remark:second} by using a symbolic programming language; we use SageMath \citep{sagemath}. \begin{remark} \label{remark:second}
The second order derivatives of $\BS_{\Put}$ are
    \begin{align*}
    \frac{\partial^2 \BS_{\Put} }{\partial x^2} (x,y) &=  \frac{1}{x\sqrt{y}}\phi(d_+(x,y)),\\
    \frac{\partial^2 \BS_{\Put} }{\partial y^2} (x,y) &= \frac{x}{4 y^{\nicefrac{3}{2}}} \phi(d_+(x,y)) \left(d_-(x,y)d_+(x,y) -1\right),\\
    \frac{\partial^2 \BS_{\Put} }{\partial y \partial x } (x,y) &= - \frac{1 }{2y}\phi(d_+(x,y))d_-(x,y)
\end{align*} for all $x,y\in\mathds{R}$ \citep{das2022closed}.
\end{remark}

Figure \ref{Fig-IG-Price} displays put option prices using the approximation formula \eqref{eq:Taylor-approx-N} when the squared volatility is an inverse Gaussian Ornstein-Uhlenbeck process. We compute from $2^{\text{nd}}$ to $6^\text{{th}}$ order approximations. We additionally display absolute error when compared to the price obtained by using the characteristic function. We note that the approximation closely resemble the option values determined by the characteristic function.

Figures \ref{Fig-IG} and \ref{Fig-IG2} show the absolute value of the error in $\log_{10}$ scale for various expiration dates, when the variance process is an inverse Gaussian Ornstein-Uhlenbeck process. The parameter $b$ in Figure \ref{Fig-IG} is valued at $80$, whereas in Figure \ref{Fig-IG2}, it is $20$. Note that when the parameter $b$ increases, the error decreases. The results in Corollary \ref{cor:IGErrorBound} are consistent with this. Similar outcomes can be achieved if the squared volatility is a gamma Ornstein-Uhlenbeck process. This is shown in Figures \ref{Fig-Gamma} and \ref{Fig-Gamma2}. Similar to the preceding instance, the error decreases with increasing parameter $b$. Observe that the inverse Gaussian and gamma model parameters are identical in Figures \ref{Fig-IG}-\ref{Fig-Gamma2}. That being said, the inverse Gaussian model yields smaller errors than the gamma model. Equations \eqref{eq:ineGamma} and \eqref{eq:ineIG} provide the reason for that: the moments of the inverse Gaussian model decrease faster with respect to $b$ than the moments of the gamma model.

We conduct numerical experiments in Figures \ref{img:IGAymptotic} and \ref{img:GammaAymptotic} to confirm the asymptotic results in Theorem \ref{Theo:Error}, for the  gamma and inverse Gaussian Ornstein-Uhlenbeck processes. When the strike reaches zero and infinity, the error converges to zero. Note again that the error associated with the gamma model are larger than for the inverse Gaussian model.

Finally, we examine the stability of put option prices as the mean reversion rate $\lambda$ increases. Figure \ref{img:Estability} compares performance of our approximation, Monte Carlo approximation \citep[p. 69]{schoutens2003levy}, and the price obtained by the characteristic function. As we can see in Figure \ref{img:Estability}, when $\lambda$ increases second and third order approximations remain stable. For large values of $\lambda$, however, the price provided by the characteristic function is unstable.We observe that the prices provided by the second and third order approximations are similar to those provided by the Monte Carlo approximation.

\begin{figure}[ht]
\begin{subfigure}{.99\textwidth} \centering
\end{subfigure}
 \caption{ Put option prices with expiration date $T=1$ and absolute value of the error when the squared volatility follows an inverse Gaussian Ornstein-Uhlenbeck process with parameters  $a=20$, $b=5$,   $\lambda = 0.5$, $\rho = -0.5$, $r=0.05$, and   $\sigma^2_0$ = 0.5.}
\label{Fig-IG-Price}
\end{figure}

\begin{figure}[ht]
\begin{subfigure}{.99\textwidth} \centering%
\end{subfigure} 
  \caption{ Absolute value of the error when the squared volatility follows an inverse Gaussian Ornstein-Uhlenbeck process with parameters  $a=20$, $b=80$,   $\lambda = 0.5$, $\rho = -0.5$, $r=0.05$, and   $\sigma^2_0$ = 0.5.}
\label{Fig-IG}
\end{figure}

\begin{figure}[ht]
\begin{subfigure}{.99\textwidth} \centering%
\end{subfigure} 
  \caption{ Absolute value of the error when the squared volatility follows an inverse Gaussian Ornstein-Uhlenbeck process with parameters  $a=20$, $b=20$,   $\lambda = 0.5$, $\rho = -0.5$, $r=0.05$, and   $\sigma^2_0 = 0.5$.}
\label{Fig-IG2}
\end{figure}

\begin{figure}[ht]
\begin{subfigure}{.99\textwidth} \centering%
\end{subfigure} 
   \caption{ Absolute value of the error when the squared volatility follows a gamma Ornstein-Uhlenbeck process  with parameters  $a=20$, $b=80$,   $\lambda = 0.5$, $\rho = -0.5$, $r=0.05$, and   $\sigma^2_0 = 0.25$.}
\label{Fig-Gamma}
\end{figure}

\begin{figure}[ht]
\begin{subfigure}{.99\textwidth} \centering%
\end{subfigure}

  \caption{ Absolute value of the error when the squared volatility follows a gamma Ornstein-Uhlenbeck process  with parameters  $a=20$, $b=20$,   $\lambda = 0.5$, $\rho = -0.5$, $r=0.05$, and   $\sigma^2_0 = 0.25$.}
\label{Fig-Gamma2}
\end{figure}

\begin{figure}[ht]
\begin{subfigure}{.45\textwidth} \centering%
\end{subfigure}
 \caption{Absolute value of the error when the squared volatility follows an inverse Gaussian Ornstein-Uhlenbeck process with parameters  $a=10$, $b=20$,   $\lambda = 0.3$, $\rho = -0.5$, $r=0.05$,  $\sigma^2_0 = 0.25$ and $S_0=100$. We fix the expiration date $T=1$ and vary the strike price.}
 \label{img:IGAymptotic}
\end{figure}

\begin{figure}[ht]
\begin{subfigure}{.45\textwidth} \centering%
\end{subfigure}
  \caption{Absolute value of the error when the squared volatility follows a gamma Ornstein-Uhlenbeck process with parameters  $a=10$, $b=20$,   $\lambda = 0.3$, $\rho = -0.5$, $r=0.05$,  $\sigma^2_0 = 0.25$ and $S_0=100$. We fix the expiration date $T=1$ and vary the strike price.}
 \label{img:GammaAymptotic}
\end{figure}

\begin{figure}[ht]
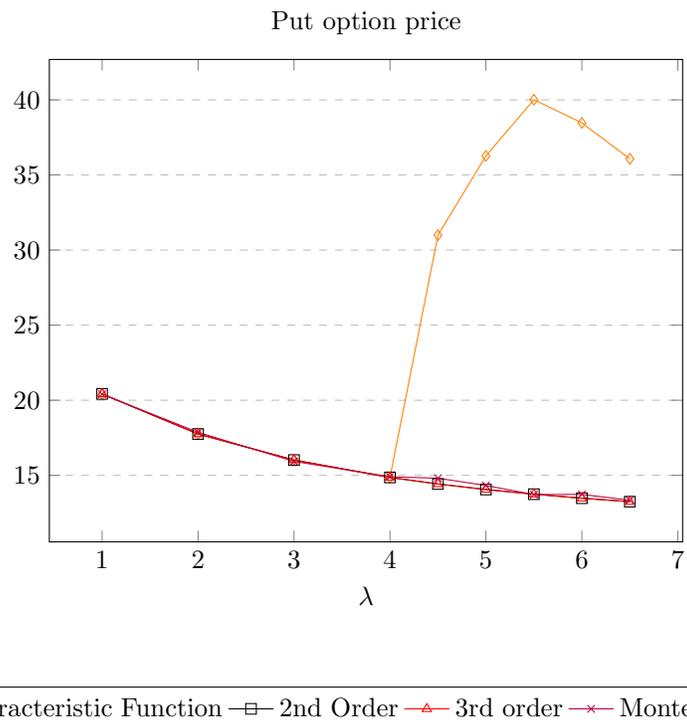

\begin{subfigure}{.99\textwidth} \centering%
\end{subfigure} 
 \caption{Put option prices with $T=1$ and $S_0=K=100$ when the squared volatility follows an inverse Gaussian Ornstein-Uhlenbeck process with parameters  $a=1$, $b=10$, $\rho = -0.3$, $r=0.05$, $\sigma^2_0 = 0.5$  and  $\lambda$ is varying.}
 \label{img:Estability}
\end{figure}

\section{Conclusion}
In this paper, we constructed approximation formulas for the Barndorff-Nielsen and Shephard model using the Romano and Touzi formula \citep{romano1997contingent}. Approximation formulas based on the Romano and Touzi formula have been implemented before for other models, like the Heston model or the GARCH diffusion model \citep{drimus2011closed,BARONEADESI2005287}. However, the authors only implemented second and third-order approximations. In our case, we can compute a closed-form approximation formula for any order.

Additionally, we performed an analysis of the error, based on the results given by \citet{das2022closed}. We extended the results of \citet{das2022closed} for $3^{\text{rd}}$ derivatives to derivatives of any order. These results showed that the error disappears for small and large values of $K$. We also gave an estimate for the error term. When $\rho=0$, the error bound for the $N^{th}$-order approximation is of the form $\mathcal{O}\left(\frac{1}{b^{N+1}}\right)$ and $\mathcal{O}\left(\frac{1}{b^{N+2}}\right)$ when the squared volatility follows a gamma Ornstein-Uhlenbeck process or an inverse Gaussian Ornstein-Uhlenbeck process, respectively. When $\rho \leq 0$, the error term can be bounded by an estimate of the form $\mathcal{O}\left(\frac{1}{b^{N+1}}\right)$ for gamma and inverse Gaussian models. As far as we know, these results have never been obtained for approximations based on the  Romano and Touzi formula. Finally, we did numerical experiments that support the theoretical results obtained in Sections \ref{SectApprox}--\ref{ErrorSect}.

This paper presented results for Barndorff-Nielsen and Shephard models in which the squared volatility follows a gamma Ornstein-Uhlenbeck process or an inverse Gaussian Ornstein-Uhlenbeck. These results can be used for other Barndorff-Nielsen and Shephard models in which the variance process follows another type of Ornstein-Uhlenbeck process, like the tempered stable Ornstein-Uhlenbeck process \citep[p.~68-70]{schoutens2003levy}. The only requirement is knowledge of the cumulant function $\kappa$ and its derivatives. For future work, this type of approximation can be used for more advanced versions of the Barndorff-Nielsen and Shephard model, such as the fractional Barndorff-Nielsen and Shephard model \citep{salmon2021fractional} or the generalized Barndorff-Nielsen and Shephard model \citep{sengupta2016generalized}. 

\begin{appendices}
\section{A bound for the derivatives of the error term}\label{DerBoundsect}

In this section, we show that high order derivatives of $\BS_{\Put}$ that appear in the error term \eqref{eq:Taylor-remainder-N} is bounded by a deterministic function when the volatility is bounded from below. First, we have the following auxiliary result.

\begin{lemma}\label{lemma:Das}
    Let $f:(0,\infty )\times [c,\infty )\to [0,\infty)$ be a continuous function with $c>0$ that satisfies
    \begin{align}
        &\lim_{y\to \infty } f(x,y) =  0 \text{ for all } x\in (0,\infty), \label{eq:PropDascond1}\\
         &\lim_{x\to 0} f(x,y)  = \lim_{x\to \infty } f(x,y) = 0 \text{ for all } y\in [c,\infty) \label{eq:PropDascond2}.
    \end{align}then the function $f$ has a  finite maximum on $(0,\infty )\times [c,\infty )$.
\end{lemma}
\begin{proof}
When $f(x,y) = 0$ for all  $(x,y) \in (0,\infty )\times [c,\infty )$, the proof is trivial and the maximum is zero. 

Assume now that there exists some $ (x_0,y_0)\in (0,\infty )\times [c,\infty )$ such that $f(x_0,y_0)>0$. Let us define the function $g: [0,\infty) \times [c,\infty) \to [0,\infty)$ as
\begin{align*}
    g(x,y)&= \begin{cases}
    f(x,y)     & \text{if }(x,y) \in (0,\infty) \times [c,\infty) \\
       0   & \text{if }(x,y) \in \{ 0 \} \times [c,\infty).
    \end{cases}
\end{align*} Observe that the function $g$ is continuous on $[0,\infty) \times [c,\infty)$. From conditions \eqref{eq:PropDascond1}--\eqref{eq:PropDascond2}  we get
    \begin{align}\label{eq:PropDasLim}
        \lim_{(x,y)\to (\infty,\infty) } g(x,y) &=0. 
    \end{align} 
    By the definition of the limit, we have that, for any $\epsilon >0$, there exists $\delta>0$ such that $|g(x,y)| < \epsilon$ for all $(x,y)\in [0,\infty)\times  [c,\infty)$   with $x^2 + y^2 > \delta^2$.
    Let us define
    $$
        \mathcal{A}_\varepsilon =\{ (x,y) \in [0,\infty) \times [c,\infty) : x^2 + y^2\leq \delta^2 \}.
    $$
    If we pick $\epsilon <f(x_0,y_0)$, then we know that if the maximum of $g$ exists then it must be in the set $\mathcal{A}_\varepsilon$. Because $\mathcal{A}_\varepsilon$ is compact and $g$ is a continuous function, we know that $g$ has a maximum on $\mathcal{A}_\varepsilon$ and hence $g$ has a finite maximum on $[0,\infty )\times [c,\infty )$. Finally, notice that
    \begin{align}\label{eq:PropdasEqual}
        \max_{(x,y)\in[0,\infty )\times [c,\infty )} g(x,y) = \max_{(x,y)\in(0,\infty )\times [c,\infty )} f(x,y).   
    \end{align} This can be proved by contradiction. If equation \eqref{eq:PropdasEqual} is not satisfied, then the maximum value of $g$ is zero, but this is not possible since $g(x_0,y_0)>0$.
 \end{proof}

Lemma \ref{lemma:Das} can now be used to bound high order derivatives of the $\BS_{\Put}$.

 \begin{proposition}\label{Prop:GeneralBound}
   For every $\mathcal{C}>0$, $\xi=(\xi_x,\xi_y)$, $\xi_x,\xi_y\in\mathds{N}_0$ such that $\lvert\xi\rvert = \xi_x+\xi_y\ge3$, there exists a function $M_{\xi,\mathcal{C}}:[0,\infty)\times(0,\infty)\rightarrow[0,\infty)$ such that
    \[
    \max_{(x,y) \in (0,\infty)\times [\mathcal{C},\infty) }\left\lvert \frac{\partial^{\lvert\xi\rvert} \BS_{\Put} }{\partial x^{\xi_x} \partial y^{\xi_y}} (x,y) \right\rvert = M_{\xi,\mathcal{C}}(T,K).
    \] The function $M_{\xi,\mathcal{C}}$ is bounded, and for fixed $K$ and $T$ its behavior is characterized by the functions $\zeta$ and $\eta$, respectively, where
    \begin{align*}
        \zeta(T) &= A_\zeta e^{-E_\zeta (rT)^2 - D_\zeta rT} \left|\sum_{k=0}^{n_\zeta} c_{k\zeta} (rT)^k \right|, \\
        \eta(K) &= A_\eta K^{-E_\eta \ln K - D_\eta} \left|\sum_{k=0}^{n_\eta} (-1)^k c_{k\eta} (\ln K)^k\right|,
    \end{align*}
    where $A_\zeta,A_\eta,D_\zeta,D_\eta\in\mathds{R}$, $n_{\zeta},n_{\eta}\in\mathds{N}_0$, $c_{0\zeta},\ldots,c_{n_\zeta\zeta},c_{0\eta},\ldots,c_{n_\eta\eta}\in\mathds{R}$ and     $E_\zeta,E_\eta>0$.
    In particular,
    \begin{align}
&\lim_{T\rightarrow0} \lvert M_{\xi,\mathcal{C}}(T,K) \rvert  < \infty  \text{ for all } K>0,\label{eq:MT02} \\
&\lim_{K\rightarrow0}M_{\xi,\mathcal{C}}(T,K)  = \lim_{K\rightarrow\infty}M_{\xi,\mathcal{C}}(T,K) = 0 \text{ for all } T>0, \label{eq:MK0Inf2}\\
&\text{if } r\neq 0, \text{ then }  \lim_{T\rightarrow\infty}M_{\xi,\mathcal{C}}(T,K) = 0 \text{ for all } K>0.\label{eq:MTInf2}
    \end{align}
\end{proposition}

Note that the functions $\zeta$ and $\eta$ depend on the lower bound $\mathcal{C}$.

\begin{proof}
    From Proposition \ref{PropDer} we can express the derivatives of $\BS_{\Put}$ as  
    \begin{align}\label{eq:TheoDasDer}
        \left\lvert \frac{\partial^{\lvert\xi\rvert} \BS_{\Put} }{\partial x^{\xi_x} \partial y^{\xi_y}} (x,y) \right\rvert & = \frac{|A_\xi|}{x^n y^{\nicefrac{m}{2}}} \phi(d_+(x,y)) \left| F_\xi (d_+(x,y), d_-(x,y), \sqrt{y})\right|
    \end{align} for $(x,y)\in (0,\infty)\times [\mathcal{C},\infty)$, where $m\in\mathds{N}$, $n\in\mathds{Z}$, $A_\xi\in\mathds{R}$ and $F_\xi$ is a polynomial in $d_+(x,y)$, $d_-(x,y)$ and $\sqrt{y}$, with degree at least one in either $d_+(x,y)$ or $d_-(x,y)$.  The arguments in the proof of Lemma 5.1 of \cite{das2022closed} apply directly in this case and allow us to obtain
    \begin{align*}
        \lim_{x\to 0} \left\lvert \frac{\partial^{\lvert\xi\rvert} \BS_{\Put} }{\partial x^{\xi_x} \partial y^{\xi_y}} (x,y) \right\rvert &=  
        \lim_{x\to \infty} \left\lvert \frac{\partial^{\lvert\xi\rvert} \BS_{\Put} }{\partial x^{\xi_x} \partial y^{\xi_y}} (x,y) \right\rvert =  0  \text{ for all } y\in [\mathcal{C},\infty),\\
        \lim_{y\to \infty} \left\lvert \frac{\partial^{\lvert\xi\rvert} \BS_{\Put} }{\partial x^{\xi_x} \partial y^{\xi_y}} (x,y) \right\rvert & =  0 \text{ for all } x\in (0,\infty).
    \end{align*} Hence equation \eqref{eq:TheoDasDer} satisfies conditions \eqref{eq:PropDascond1}--\eqref{eq:PropDascond2}, and we can apply Lemma \ref{lemma:Das}. This means that the absolute value of the derivative of $\BS_{\Put}$ has a maximum on $(0,\infty)\times[\mathcal{C},\infty)$. There exists a point $(x_*,y_*) \in (0,\infty)\times[\mathcal{C},\infty)$ that satisfies\begin{align}
    &\max_{(x,y) \in (0,\infty)\times [\mathcal{C},\infty) }\left\lvert \frac{\partial^{\lvert\xi\rvert} \BS_{\Put} }{\partial x^{\xi_x} \partial y^{\xi_y}} (x,y) \right\rvert\nonumber \\
    &\quad  = \frac{|A_\xi|}{x_*^n y_*^{\nicefrac{m}{2}}} \phi(d_+(x_*,y_*)) \left| F_\xi (d_+(x_*,y_*), d_-(x_*,y_*), \sqrt{y_*})\right|.\label{eq:MaxDer}
    \end{align}
Proceeding as in the proof of Lemma 5.2 of \cite{das2022closed}, if we fix all the variables, except the expiry date $T$,  equation \eqref{eq:MaxDer} can be expressed as
        \begin{align*}
        \zeta(T) &= A_\zeta e^{-E_\zeta (rT)^2 - D_\zeta rT} \left|\sum_{k=0}^{n_\zeta} c_{k\zeta} (rT)^k\right|,
    \end{align*}  where $D_\zeta\in\mathds{R}$, $n_{\zeta}\in\mathds{N}_0$, $c_{0\zeta},\ldots,c_{n_\zeta\zeta}\in\mathds{R}$ and     $A_\zeta, E_\zeta>0$. It is clear that when $T\to\infty$, $\zeta$ goes to zero, and when $T\to 0$, the limit of $\zeta$ is $A_\zeta |c_{0\zeta}|$. Similarly, if we fix all the variables except the strike price $K$, then \eqref{eq:MaxDer} becomes
        \begin{align*}
        \eta(K) &= A_\eta K^{-E_\eta \ln K - D_\eta} \left|\sum_{k=0}^{n_\eta} (-1)^k c_{k\eta} (\ln K)^k\right|,
    \end{align*}
    where $D_\eta\in\mathds{R}$, $n_{\eta}\in\mathds{N}_0$, $c_{0\eta},\ldots,c_{n_\eta\eta}\in\mathds{R}$ and     $A_\eta,E_\eta>0$. The function $\eta$ goes to zero when $K\to \infty$ or when $K\to 0$. \end{proof}


We now consider how the maximum will behave when the volatility can reach the value zero. 

\begin{proposition}\label{Prop:ZeroBound}
  The function $M_{\xi,\mathcal{C}}:[0,\infty)\times(0,\infty)\rightarrow[0,\infty)$ defined in Proposition \ref{Prop:GeneralBound} satisfies \begin{equation} \label{eq:limit0}
      \lim_{\mathcal{C}\to 0 } M_{\xi,\mathcal{C}}(T,K)  = \infty  \text{ for all } T,K>0.
  \end{equation}
\end{proposition}
\begin{proof}
    By application of Proposition \ref{PropDer} and the proof of Lemma 5.1 of \citet{das2022closed} we have that \begin{equation}\label{eq:LimInft}
          \lim_{y\to 0} \left\lvert \frac{\partial^{\lvert\xi\rvert} \BS_{\Put} }{\partial x^{\xi_x} \partial y^{\xi_y}} (x,y) \right\rvert  =\begin{cases}
              \infty  & \text{when } \ln \frac{x}{K} + rT =0, \\
              0 &  \text{when } \ln \frac{x}{K} + rT \neq 0.
          \end{cases}
    \end{equation} From equation \eqref{eq:LimInft} it is clear that the supremum on the set $(0,\infty)\times[0,\infty)$ is reached at $\left(x_*,y_*\right) =  \left(Ke^{-rT}, 0\right)$ and hence \eqref{eq:limit0} is satisfied.
\end{proof}
    
The next result comes from the fact that $Z$ is an L\'evy subordinator. Since $Z$ is a non-negative and a non-decreasing process, $\alpha_{s,T}\geq 0$ for all $s\in[0,T]$ and by equation \eqref{eq:I-SDE}, we have the lower bounds
$$ I_T = \sigma^2_0 \alpha_{0,T} +  \int_0^T \alpha_{s,T} dZ_{\lambda s} \geq \sigma^2_0 \alpha_{0,T} >0 $$
for all $T>0$ and
\begin{equation}
      (1-u)E_{\mathds{Q}}[I_T]+uI_T  \geq \sigma^2_0 \alpha_{0,T}\label{eq:JLowerBound} 
\end{equation} for all $u \in(0,1)$ and every $T>0$, with both bounds holding almost surely. Hence, we obtain the following result from \eqref{eq:JLowerBound} and Proposition \ref{Prop:GeneralBound}.
\begin{proposition}\label{cor:boundsDer}
    For every $\xi=(\xi_x,\xi_y)$, $\xi_x,\xi_y\in\mathds{N}_0$ such that $\lvert\alpha\rvert = \xi_x+\xi_y\ge3$, there exists a function $M_{\xi,\beta(T)}:[0,\infty)\times(0,\infty)\rightarrow[0,\infty)$ such that
    \[
    \left\lvert \frac{\partial^{\lvert\xi\rvert} \BS_{\Put} }{\partial x^{\xi_x} \partial y^{\xi_y}} ((1-u)S_0+uP_TS_0,(1-u)E_{\mathds{Q}}[I_T]+uI_T) \right\rvert \le M_{\xi,\beta(T)}(T,K)
    \]  
    for all $u\in(0,1)$ almost surely, where $M_{\xi,\beta(T)}$ is defined as in Proposition \ref{Prop:GeneralBound} and $\beta(T) = \sigma^2_0 \alpha_{0,T}$.  In addition,
    \begin{align}
    &\lim_{T\rightarrow0}M_{\xi,\beta(T)} (T,K)   = \infty  \text{ for all } K>0,\label{eq:MT03} \\
    &\lim_{K\rightarrow0}M_{\xi,\beta(T)}(T,K)  = \lim_{K\rightarrow\infty}M_{\xi,\beta(T)}(T,K) = 0 \text{ for all } T>0, \label{eq:MK0Inf3}\\
        &\text{if } r\neq 0, \text{ then } \lim_{T\rightarrow\infty}M_{\xi,\beta(T)}(T,K) = 0 \text{ for all } K>0.\label{eq:MTInf3}
    \end{align}
\end{proposition}\begin{proof}
    If we fix a expiry date $T>0$, then from Proposition \ref{Prop:GeneralBound} and \eqref{eq:JLowerBound}, we have that
        \[
    \sup_{u\in(0,1)}\left\lvert \frac{\partial^{\lvert\xi\rvert} \BS_{\Put} }{\partial x^{\xi_x} \partial y^{\xi_y}} ((1-u)S_0+uP_TS_0,(1-u)E_{\mathds{Q}}[I_T]+uI_T) \right\rvert = M_{\xi,\beta(T)}(T,K),
    \] almost surely. The result in \eqref{eq:MK0Inf3} comes from \eqref{eq:MK0Inf2}. Observe that the lower bound of $(1-u)E_{\mathds{Q}}[I_T]+uI_T$ depends on the expiry date $T$. This means that the limits of $M_{\xi,\beta(T)}$ when $T\to 0$ and $T\to \infty$, have to be treated carefully. When $T \to \infty$, we have that
    $$
       \lim_{T\to \infty} \beta(T)  = \frac{\sigma^2_0}{\lambda}>0.
    $$
    Since this limit is bounded away from zero, from Proposition \ref{Prop:GeneralBound} we obtain the result in \eqref{eq:MTInf3}. The result in \eqref{eq:MT03} is satisfied due to Proposition \ref{Prop:ZeroBound} because
    \begin{equation*}
       \lim_{T\to 0} \beta(T)  = 0. \qedhere
    \end{equation*}
\end{proof}

Proposition \ref{cor:boundsDer} allows us to construct deterministic bounds for the derivatives that appear in the remainder term \eqref{eq:Taylor-remainder-N}. These bounds are finite for every $T,K>0$ and are well-behaved for small values of $K$ and for large values of $K$ and $T$. However, these bounds are large when $T$ is near zero.

\section{Characteristic functions}\label{CFsect}
In this section, we present the characteristic functions of the log price that are used in Section \ref{NumExamplesSect}. We use the characteristic function of the log price to compute the values of European put prices using the formula given by \citet{carr1999option}. We assume that the options prices given by the characteristic function of the log price, represent the true price of options. In Section \ref{NumExamplesSect}, we have compared the option price given by the characteristic function with the approximation method. For a general Barndorff-Nielsen and Shephard model, the characteristic function of the log price at time $T>0$ is:
\begin{align}
E\left[e^{iuX_T}\right]  &=  e^{iuX_0 +iurT -iu\lambda\kappa(\rho) T - \frac{1}{2}(iu+u^2) \frac{\sigma_0^2}{\lambda} \left( 1-e^{-\lambda T}\right)} \nonumber\\
& e^{\lambda \int_0^T \kappa\left( iu \rho - \frac{1}{2} (iu+u^2) \frac{1}{\lambda} \left( 1- e^{-\lambda(T-s)}\right)\right) ds}, \label{CF_alleq}
\end{align}
where $\kappa$ is the cumulant generating function of $Z_1$ \citep{nicolato2003option}. We observe that the characteristic function of $X_T$ depends on the type of $D$-Ornstein-Uhlenbeck process we use for modelling the variance process. In this paper, we focus only on two types of $D$-Ornstein-Uhlenbeck process. From equation \eqref{CF_alleq} it is clear that all Barndorff-Nielsen and Shephard models have the same characteristic function except for the Riemann integral that appears in equation  \eqref{CF_alleq}. The value of this integral depends on the type of $D$-Ornstein-Uhlenbeck process.
\begin{enumerate}
    \item  When the variance process $\sigma^2$ follows an inverse Gaussian\break Ornstein-Uhlenbeck process, the integral that appears in equation \eqref{CF_alleq} satisfies the following equation
\begin{align}
&\lambda \int_0^T \kappa\left( iu \rho - \frac{1}{2} (iu+u^2) \alpha_{s,T}\right) ds  \nonumber\\
   & \qquad =    a \left(\sqrt{b^2 -2f_1(u)} -\sqrt{b^2-2iu\rho}\right) + \frac{2a f_2(u)}{\sqrt{2f_2(u)-b^2}} \nonumber\\
   & \qquad \qquad \left[ \arctan\left( \sqrt{\frac{b^2-2iu\rho}{2f_2(u)-b^2}}\right) - \arctan\left(\sqrt{\frac{b^2 - 2f_1(u)}{2f_2(u)-b^2}}\right)\right],\label{CFIG}
\end{align} where
\begin{align}
     f_1(u) & = iu\rho -\frac{1}{2 \lambda}(u^2 + iu)\left(1-e^{-\lambda T}\right), \label{f1eq}\\
    f_2(u) & = iu\rho -\frac{1}{2 \lambda}(u^2 + iu), \label{f2eq}
\end{align} \citep{nicolato2003option}.
\item In the case the process $\sigma^2$ is a gamma Ornstein-Uhlenbeck process, the integral of equation \eqref{CF_alleq} can be written as
\begin{multline}
\lambda \int_0^T \kappa\left( iu \rho - \frac{1}{2} (iu+u^2) \alpha_{s,T}\right) ds \\
  \begin{aligned}\label{CFGamma}
  =&    a(b-f_2(u))^{-1} \left(b\ln\left(\frac{b-f_1(u)}{b-iu\rho}\right) + f_2(u)\lambda T \right),
  \end{aligned}
\end{multline} where $f_1$ and $f_2$ are defined as in equations \eqref{f1eq} and \eqref{f2eq} respectively \citep{nicolato2003option}.
\end{enumerate}

\begin{remark}
We detect a small typo in the characteristic functions of Barndorff-Nielsen and Shephard models given by \citet{nicolato2003option} and by \citet[p.~87]{schoutens2003levy}. This typo is related to the functions $f_1$ and $f_2$ defined in equations \eqref{f1eq} and \eqref{f2eq} respectively. In the formula given by  \citet{nicolato2003option} and \citet[p.~87]{schoutens2003levy}, the factor $\frac{1}{\lambda}$ appears to be missing in $f_1$ and $f_2$. For example, to obtain equation \eqref{CFGamma} it is possible to show that
\begin{multline}
\int_0^T \kappa\left( iu \rho - \frac{1}{2} (iu+u^2) \frac{1}{\lambda} \left( 1- e^{-\lambda(T-s)}\right)\right) ds \label{IntegralGamma}
\\
  \begin{aligned}
  =&  \left.\frac{2a}{ \lambda  u^{2}+\left(-2 i \lambda^{2} \rho +i \lambda \right) u+2 \lambda^{2} b} F_\kappa(s) \right|_0^T,
  \end{aligned} 
\end{multline}where
\begin{multline*}
    F_\kappa(s) = - \ln\! \left(u \left(i+u\right) {\mathrm e}^{-\lambda \left(T-s\right)}-u^{2}+i \left(2 \lambda  \rho -1\right) u-2 b \lambda \right) b \lambda \\
    + u \ln\! \left({\mathrm e}^{-\lambda  \left(T-s\right)}\right) \left(i \rho  \lambda -\frac{1}{2} i-\frac{1}{2} u\right).
\end{multline*}
Notice that the denominator in equation \eqref{IntegralGamma} can be written as
\[\lambda  u^{2}+\left(-2 i \lambda^{2} \rho +i \lambda \right) u+2 \lambda^{2} b  = 2\lambda^2 (b-f_2(u)).\]
For the numerator, we need to evaluate the function $F_\kappa$ at $s=T$ and at $s=0$. For $F_\kappa(T)$ we obtain \[  F_\kappa(T)  = -\ln\left(\left(\rho  i u-b\right) 2\lambda \right) \] and in the case of $F_\kappa(0)$ we have
\begin{align*}
  F_\kappa(0)
  & = -\ln\left(2\lambda \left(f_1(u) -b \right) \right) b\lambda  - \lambda^2 T f_2(u).
\end{align*}
So from equation \eqref{IntegralGamma} we have that
\begin{multline*}
\int_0^T \kappa\left( iu \rho - \frac{1}{2} (iu+u^2) \frac{1}{\lambda} \left( 1- e^{-\lambda(T-s)}\right)\right) ds
\\
  \begin{aligned}
  &=  \frac{2a}{2\lambda^2 (b-f_2(u))} \left( F_\kappa(T) - F_\kappa(0) \right)\\
  &= \frac{a}{\lambda^2 (b-f_2(u))} \left( -\ln\left(\left(\rho  i u-b\right) 2\lambda \right) +\ln\left(2\lambda \left(f_1(u) -b \right) \right) b\lambda + \lambda^2 T f_2(u) \right)\\
  &=  \frac{a}{\lambda (b-f_2(u))} \left( b\ln\left(\frac{b-f_1(u)}{b-iu\rho}\right)  + \lambda Tf_2(u) \right),
  \end{aligned}
\end{multline*} as required. It is also possible to show the result that appear in equation \eqref{CFIG}, but the proof is quite tedious and it is omitted here.
\end{remark}

\end{appendices}

\bibliography{ref}

@book{pascucci2011pde,
  title={{PDE} and martingale methods in option pricing},
  author={Pascucci, Andrea},
  year={2011},
  publisher={Springer}
}

@article{barndorff2003integrated,
  title={Integrated {O}rnstein--{U}hlenbeck processes and non-{G}aussian {O}rnstein--{U}hlenbeck based Stochastic volatility models},
  author={Barndorff-Nielsen, Ole E and Shephard, Neil},
  journal={Scandinavian Journal of statistics},
  volume={30},
  number={2},
  pages={277--295},
  year={2003}
}

@article{das2022closed,
  title={Closed-form approximations with respect to the mixing solution for option pricing under stochastic volatility},
  author={Das, Kaustav and Langren{\'e}, Nicolas},
  journal={Stochastics},
  volume={94},
  number={5},
  pages={745--788},
  year={2022}
}

@article{romano1997contingent,
  title={Contingent claims and market completeness in a stochastic volatility model},
  author={Romano, Marc and Touzi, Nizar},
  journal={Mathematical Finance},
  volume={7},
  number={4},
  pages={399--412},
  year={1997}
}

@article{hull1987pricing,
  title={The pricing of options on assets with stochastic volatilities},
  author={Hull, John and White, Alan},
  journal={The Journal of Finance},
  volume={42},
  number={2},
  pages={281--300},
  year={1987}
}

@Inbook{Duistermaat2010,
author="Duistermaat, J. J.
and Kolk, J. A. C.",
title="Taylor Expansion in Several Variables",
bookTitle="Distributions: Theory and Applications",
year="2010",
publisher="Birkh{\"a}user",
address="Boston",
pages="59--63",
isbn="978-0-8176-4675-2",
doi="10.1007/978-0-8176-4675-2_6"}

@article{nicolato2003option,
  Title                    = {Option Pricing in Stochastic Volatility Models of the {O}rnstein-{U}hlenbeck type},
  Author                   = {Elisa Nicolato and Emmanouil Venardos},
  Journal                  = {Mathematical Finance},
  Year                     = {2003},
  Number                   = {4},
  Pages                    = {445--466},
  Volume                   = {13}
}

@manual{sagemath,
  Key          = {SageMath},
  Author       = {{The Sage Developers}},
  Title        = {{S}ageMath, the {S}age {M}athematics {S}oftware {S}ystem ({V}ersion 9.7)},
  note         = {{\tt https://www.sagemath.org}},
  Year         = {2023},
}

@article{carr1999option,
  title={{O}ption valuation using the fast {F}ourier transform},
  author={ Carr, Peter and Madan, Dilip B.},
  journal={Journal of Computational Finance},
  volume={2},
  number={4},
  pages={61--73},
  year={1999}
}

@InBook{Barndorff-Nielsen_Shephard2001b,
  author    = {Barndorff-Nielsen, Ole E. and Shephard, Neil},
  editor    = {Barndorff-Nielsen, Ole E. and Resnick, Sidney I. and Mikosch, Thomas},
  pages     = {283--318},
  publisher = {Birkh{\"a}user Boston},
  title     = {Modelling by L{\'e}vy Processess for Financial Econometrics},
  year      = {2001},
  address   = {Boston, MA},
  isbn      = {978-1-4612-0197-7},
  booktitle = {L{\'e}vy Processes: Theory and Applications},
}

@Article{Barndorff-Nielsen_Shephard2001a,
  author   = {Ole E. Barndorff-Nielsen and Neil Shephard},
  journal  = {Journal of the Royal Statistical Society, Series B (Statistical Methodology)},
  title    = {Non-{G}aussian {O}rnstein-{U}hlenbeck-Based Models and Some of Their Uses in Financial Economics},
  year     = {2001},
  number   = {2},
  pages    = {167--241},
  volume   = {63},
}

@article{joshi2001log,
  title={Log-type models, homogeneity of option prices and convexity},
  author={Joshi, M},
  journal={QUARC, Royal Bank of Scotland working paper},
  year={2001},
  publisher={Citeseer}
}

@article{hubalek2011joint,
  title={Joint analysis and estimation of stock prices and trading volume in {B}arndorff-{N}ielsen and {S}hephard stochastic volatility models},
  author={Hubalek, Friedrich and Posedel, Petra},
  journal={Quantitative Finance},
  volume={11},
  number={6},
  pages={917--932},
  year={2011},
  publisher={Taylor \& Francis}
}

@PhdThesis{GuineaJulia2022,
  author = {{\'A}lvaro {Guinea Juli{\'a}}},
  school = {University of York},
  title  = {Incorporating market attention in option pricing with applications to {B}itcoin derivatives},
  year   = {2022},
}

@article{arai2022approximate,
  title={APPROXIMATE OPTION PRICING FORMULA for {B}ARNDORFF-{N}IELSEN and {S}HEPHARD MODEL},
  author={Arai, Takuji},
  journal={International Journal of Theoretical and Applied Finance},
  volume={25},
  number={2},
  pages={2250008},
  year={2022},
  publisher={World Scientific Publishing Co. Pte Ltd}
}

@article{alos2012decomposition,
  title={A decomposition formula for option prices in the {H}eston model and applications to option pricing approximation},
  author={Al{\`o}s, Elisa},
  journal={Finance and Stochastics},
  volume={16},
  pages={403--422},
  year={2012},
  publisher={Springer}
}

@article{schroder2014contingent,
  title={On contingent-claim valuation in continuous-time for volatility models of {O}rnstein--{U}hlenbeck type},
  author={Schr{\"o}der, Michael},
  journal={Journal of computational and applied mathematics},
  volume={260},
  pages={36--53},
  year={2014},
  publisher={Elsevier}
}

@book{schoutens2003levy,
  title={{L}{\'e}vy processes in finance: pricing financial derivatives},
  author={Schoutens, Wim},
    publisher={John Wiley \& Sons},
  year={2003}
}

@book{cont2004financial,
  title={{Financial modelling with jump processes}},
  author={Cont, Rama and Tankov, Peter},
  year={2004},
  publisher={Chapman \& Hall/CRC}
}

@book{sato1999levy,
  title={{L}{\'e}vy processes and infinitely divisible distributions},
  author={Sato, {Ken-iti}},
  year={1999},
  publisher={Cambridge University Press}
}

@article{valdivieso2009maximum,
  title={Maximum likelihood estimation in processes of {O}rnstein--{U}hlenbeck type},
  author={Valdivieso, Luis and Schoutens, Wim and Tuerlinckx, Francis},
  journal={Statistical Inference for Stochastic Processes},
  volume={12},
  number={1},
  pages={1--19},
  year={2009},
  publisher={Springer}
}

@article{sengupta2016generalized,
  title={Generalized {BN--S} stochastic volatility model for option pricing},
  author={SenGupta, Indranil},
  journal={International Journal of Theoretical and Applied Finance},
  volume={19},
  number={02},
  pages={1650014},
  year={2016},
  publisher={World Scientific}
}

@article{salmon2021fractional,
  title={Fractional {B}arndorff-{N}ielsen and {S}hephard model: applications in variance and volatility swaps, and hedging},
  author={Salmon, Nicholas and SenGupta, Indranil},
  journal={Annals of Finance},
  volume={17},
  number={4},
  pages={529--558},
  year={2021},
  publisher={Springer}
}

@article{drimus2011closed,
  title={Closed-form convexity and cross-convexity adjustments for {H}eston prices},
  author={Drimus, Gabriel G},
  journal={Quantitative Finance},
  volume={11},
  number={8},
  pages={1137--1149},
  year={2011},
  publisher={Taylor \& Francis}
}

@article{heston1993closed,
  title={{A} closed-form solution for options with stochastic volatility with applications to bond and currency options},
  author={Heston, Steven L},
  journal={The Review of Financial Studies},
  volume={6},
  number={2},
  pages={327--343},
  year={1993},
  publisher={Oxford University Press}
}

@article{BARONEADESI2005287,
title = {An option pricing formula for the GARCH diffusion model},
journal = {Computational Statistics \& Data Analysis},
volume = {49},
number = {2},
pages = {287-310},
year = {2005},
note = {2nd CSDA Special Issue on Computational Econometrics},
issn = {0167-9473},
doi = {https://doi.org/10.1016/j.csda.2004.05.014},
author = {Giovanni Barone-Adesi and Henrik Rasmussen and Claudia Ravanelli}
}

@article{Das2024Erratum,
author = {Kaustav Das and Nicolas Langrené},
title = {Erratum for ‘Closed-form approximations with respect to the mixing solution for option pricing under stochastic volatility’},
journal = {Stochastics},
volume = {0},
number = {0},
pages = {1--4},
year = {2024},
publisher = {Taylor \& Francis},
doi = {10.1080/17442508.2024.2329129}}

\end{document}